\documentclass[12]{amsart}

\usepackage{amssymb}
\makeatletter
\numberwithin{equation}{section}
\numberwithin{figure}{section}
\numberwithin{table}{section}
\newtheorem{thm}{Theorem}[section]
\newtheorem{lemma}[thm]{Lemma}
\newtheorem{cor}[thm]{Corollary}
\newtheorem{pro}[thm]{Proposition}
\newtheorem{defn}[thm]{Definition}

\newtheorem{rem}[thm]{Remark}

\begin{document}
\title[From Braces to Hecke  algebras $\&$ Quantum Groups]{ From Braces to Hecke  algebras $\&$ Quantum Groups}
\author[Anastasia Doikou and Agata Smoktunowicz]{Anastasia Doikou and Agata Smoktunowicz}

\address[A. Doikou] {Department of Mathematics, Heriot-Watt University,
Edinburgh EH14 4AS, and Maxwell Institute for Mathematical Sciences, Edinburgh}
\email{A.Doikou@hw.ac.uk}

\address[A. Smoktunowicz] {School of Mathematics, The University of Edinburgh, 
The Kings Buildings, Mayfield Road, Edinburgh  EH9 3JZ,
and Maxwell Intitute for Mathematical Sciences,  Edinburgh}
\email{A.Smoktunowicz@ed.ac.uk}

%\subjclass[2010]{Primary 16N80, 16P90, 16N40 } 
%\keywords{brace, Yang-Baxter equation, nilpotent braces}

\date{\today}

\begin{abstract}
We examine links  between the theory of braces and set theoretical solutions of the Yang-Baxter equation, 
and fundamental concepts from the theory of quantum integrable systems.   
More precisely, we make connections with Hecke algebras and  we identify new quantum groups associated 
to set-theoretic solutions coming from braces. We also construct a novel class of quantum discrete integrable systems 
and we derive symmetries for the corresponding periodic transfer matrices. 

$ $


\end{abstract}

\maketitle

\section{Introduction}
\noindent The Yang-Baxter equation  is a fundamental equation
in the theory of quantum integrable models and solvable statistical systems, as well as
in the formulation of quantum groups \cite{FadTakRes, Jimbo, Drinfeld}. It was introduced in \cite{Yang}
as a main tool for the investigation of many particle systems with  $\delta$-type  interactions, and in \cite{Baxter} 
for the study of a two-dimensional solvable statistical model.
Since Drinfeld  \cite{Drin}  suggested  a theory of  set-theoretic solutions to the Yang-Baxter equation be developed,  
set-theoretic solutions have been extensively  investigated using braided groups, and more recently by applying 
the theory of  braces and skew-braces. 

Set theoretical solutions and Yang-Baxter maps have 
been also extensively studied in the context of classical  discrete integrable systems linked to the notion of  
Darboux-B{\"a}cklund transformation within the Lax pair formulation  \cite{ABS, Veselov, Papag}. 
In classical integrable systems usually a Poisson structure exists associated to a classical $r$-matrix, 
which is a solution of the classical Yang-Baxter equation \cite{FT}. Also, relevent recent results on Yang-Baxter maps,
when the quantum group symmetry  is  a priori requirement can be found in \cite{Baz}.
 
It is worth noting that \cite{Hienta} provides one of the first instances of classification of
set-theoretical solutions of Yang-Baxter equation. 
Various  connections between the set theoretical Yang-Baxter equation and geometric crystals \cite{Eti, BerKaz}, 
or soliton cellular automatons \cite{TakSat, HatKunTak}  have been also demonstrated.

The theory of braces  was established around 2005, when Wolfgang Rump developed a structure called a brace  to describe 
all finite involutive  set-theoretic solutions of  the Yang-Baxter equation.   
Rump showed that every brace yields a solution to the Yang-Baxter equation, and every 
non-degenerate, involutive set-theoretic  solution of the Yang-Baxter equation can be obtained from a brace, 
a structure which generalises nilpotent rings.  
Subsequently skew-braces were developed by  Guarnieri and Vendramin  to describe non-involutive solutions \cite{GV}. 

The theory of braces and skew braces has connections with numerous research areas,
for example with  group theory (Garside groups, regular subgroups, factorised groups-- see for example \cite{bcjo, JKA, kava, sysak}), 
algebraic number theory, Hopf-Galois extensions \cite{Ba, SVB}, non-commutative ring theory  \cite{[26], KSV, IL},  Knot theory \cite{Ka, Lebed}, 
Hopf algebras, quantum groups \cite{ESS}, universal algebra, groupoids \cite{APil}, semi-braces \cite{Catino}, 
trusses \cite{Tomasz} and Yang-Baxter maps. 
Moreover, skew braces are related to non-commutative physics, Yetter-Drinfield modules 
and Nichols algebras.

Note that in the present article when we say set theoretic solutions we always mean  finite, non-degenerate, involutive, 
set-theoretic solutions of the Yang-Baxter equation. As was shown by Rump all finite, non-degenerate, involutive, 
set-theoretic solutions of the Yang-Baxter equation (\ref{brace1}) are coming from braces (Theorem \ref{Rump}),
therefore some times we may call such solutions {\it  brace solutions}.
Because set-theoretic solutions coming from braces are involutive, it is possible to {\it Baxterise} them \cite{Jimbo2} and obtain 
solutions to the parameter dependent Yang-Baxter equation, which appear in quantum integrable systems.  

The aim of this paper is to investigate connections between the theory of braces and selected topics 
from the theory of quantum integrable systems.  More precisely:
\begin{enumerate}
\item We derive new quantum groups associated to braces.
\item We construct a novel class of quantum discrete  integrable systems.
\item We identify symmetries of the periodic transfer matrices of the novel integrable systems.
\end{enumerate}

Note that in \cite{ESS} Etingof, Schedler and Soloviev constructed quantum groups associated to 
set-theoretic solutions, however we use a different construction coming from parameter 
dependent solutions of the Yang-Baxter equation and  our quantum groups differ from these in \cite{ESS}.

Rump showed that  every nilpotent ring is a brace, therefore  readers who are not familiar with the theory of 
braces may replace use of the word brace in this paper with the words nilpotent ring. Readers interested in 
learning more about the  theory of braces are referred to \cite {fc, [6], gateva, [25], [26], LAA}.

\subsection*{Structure of the paper} 
This paper is divided into four sections:

\begin{enumerate}
\item 
 Section $1$ contains the  introduction of the paper.
\item
 Section $2$ shows how to construct $R$-matrices associated to non-degenerate, involutive, 
set-theoretic solutions of the Yang-Baxter equation in preparation for Section $3$.

\item
 Section $3$ contains information on how to go about connecting the theory of quantum integrable systems with 
$R$-matrices constructed from braces as per Section $2$. We construct the new quantum algebra associated to braces. 
We also construct various realizations of the relevant quantum algebras 
using classical results from the theory of braces, along with posing some more open questions. Note that the Yangian
is a special case within the larger class of quantum algebras emerging from braces.

The section consists of three subsections:
\begin{itemize}
\item 3.1: The Yang-Baxter equation $\&$ $A$-type Hecke algebra.
\item 3.2: Quantum algebras from braces.
\item 3.3: Representations of quantum algebras.
\end{itemize}
At the end of each section there are relevant questions and lines of inquiry for further research.

\item Section $4$, offers information on the construction of a {\it new class} of integrable quantum spin chain-like systems 
associated to braces. Spin chain-like systems are typically constructed  by means of tensor realizations of the underlying 
quantum algebra, by introducing the so called transfer matrix. We show that the periodic transfer matrix constructed 
from Baxterized solutions of the $A$-type Hecke algebra ${\mathcal H}_N(q=1)$ can be exclusively expressed in terms of the generators of the 
$A$-type Hecke algebras plus some periodic term. This is 
a universal result that holds for any representation of the $A$-type Hecke algebra 
${\mathcal H}_N(q=1),$ and is one of the main propositions of this study. 
We then focus on  solutions of the Yang-Baxter equation coming form braces and we
employ the theory of braces to construct symmetries of the corresponding periodic transfer matrices.

The  section consists of two subsections:
\begin{itemize}
\item 4.1: Tensor representations of quantum algebras $\&$ integrable systems.
\item 4.2: Symmetries of the periodic transfer matrix of novel classes of spin chains.

\end{itemize}

\end{enumerate}

\section{Basic information about braces $\&$ set-theoretic solutions}

\noindent {\em For a set-theoretic solution of the Braid equation, we will use notation $(X, {\check r})$, instead of the usual notation  $(X,r)$, 
 to be consistent with notations used in quantum integrable systems.}

Let ${\mathcal N}$ be a natural number, and 
 let $e_{i,j}$ denote the ${\mathcal N}\times {\mathcal N}$ matrix whose all entries are $0$ except for the $i,j$-th entry, which is $1$.

Let $X=\{x_{1}, \ldots x_{{\mathcal N}}\}$ be a set and ${\check r}:X\times X\rightarrow X\times X$.
 Denote \[{\check r}(x,y)=(\sigma _{x}(y), \tau _{y}(x))=({{ }^xy}, x^{y}).\] We say that $\check r$ is non-degenerate if $\sigma _{x}$ and $\tau _{y}$ 
are bijective functions. 
Suppose that $(X, {\check r})$ is an involutive, non-degenerate set-theoretic solution of the Braid equation:
\[({\check r}\times Id_{X})(Id_{X}\times {\check r})({\check r}\times Id_{X})=(Id_{X}\times {\check r})({\check r}\times Id_{X})(Id_{X}\times {\check r}).\]

With a slight abuse of notation, let  $\check r$ also  denote the $R$-matrix associated to the linearisation of ${\check r}$ on 
$V={\mathbb C }X$ (see \cite{LAA} for more details). 

This matrix is also called {\em a check-matrix}. 
Then the check-matrix related  to $(X,{\check r})$ is ${\check r}={\check r}_{i,j; k,l}$, where ${\check r}_{i,j; k,l}=1 $ if and only if ${\check r}(i,j)=(k,l)$, 
and is zero otherwise.
 Notice that the matrix $\check r:V\otimes V\rightarrow V\otimes V$ satisfies the (constant) Braid equation:
\[({\check r}\otimes Id_{V})(Id_{V}\otimes {\check r})({\check r}\otimes Id_{V})=(Id_{V}\otimes {\check r})({\check r}\times 
Id_{V})(Id_{V}\otimes {\check r}).\]
 Notice that ${\check r}^{2}=I_{V \otimes V}$ the identity matrix, because $\check r$ is involutive.

Let $r=\tau  {\check r}$ be the corresponding solution of the quantum Yang-Baxter equation (where $\tau (x,y)=(y,x)$)
and let  $r$ denote the matrix associated to the linearisation of $r$ on $V=\mathbb C X$. Notice that the matrix $ r$ satisfies
the constant Yang-Baxter equation.

 Let $e_{x, y}$ be the matrix with $x, y$ entry equal to $1$ and all the other entries $0$.
\begin{lemma}\label{1}
 Let notation be as above. Then  the matrix  ${\check r}$  has the form:
\begin{equation}
{\check r}=\sum_{x, y\in X} e_{x, \sigma_x(y)}\otimes e_{y, \tau_y(x)}. \label{brace1}
\end{equation}
 The matrix $r$ has the form: 
\begin{equation}
r={\mathcal P}\cdot {\check r}= (\sum_{x,y\in X} e_{y,x}\otimes e_{x,y})(\sum_{x, y\in X} e_{x, \sigma_x(y)}\otimes e_{y, \tau_y(x)}), \nonumber
\end{equation}
 consequently
\begin{equation}
r=\sum_{x, y\in X} e_{y, \sigma_x(y)}\otimes e_{x, \tau_y(x)}. \label{brace2}
\end{equation}

 Moreover because $\check r$ is involutive we get ${\check r}(\sigma_x(y), \tau_y(x))=(x,y)$, therefore 
\[r=\sum_{x, y \in X} e_{\tau_y(x), x}\otimes e_{\sigma_x(y), y}\]
\end{lemma}
\begin{proof} We use a direct calculation using the way in which the $R$-matrix associated to a set-theoretic solution $(X, {\check r})$ 
is built (for more details, 
see Definition 2.3 in  \cite{LAA}). It is worth noticing that,  in some books on quantum groups, the obtained check-matrix is also transposed. However, 
for involutive solutions, which we consider in this paper,  the obtained $R$-matrix is symmetric, so it is the same.
\end{proof}

 In \cite{[25], [26]} Rump showed that every solution $(X,r)$ can be in a good way embedded in a brace.
 \begin{defn}[Proposition $4$, \cite{[26]}]
 A {\em left brace } is an abelian group $(A; +)$ together with a 
 multiplication $\cdot $ such that the circle operation $a \circ  b =
 a\cdot b+a+b$ makes $A$ into a group, and $a\cdot (b+c)=a\cdot b+a\cdot c$.
\end{defn} 
 In many papers, the following equivalent definition from \cite{[6]} is used:
 \begin{defn}[\cite{[6]}]
 A {\em left brace } is a set $G$ together with binary operations $+$ and $\circ $ such that 
 $(G, +)$ is an abelian group, $(G, \circ )$ is a group, and 
  $a\circ (b+c)+a=a\circ b+a\circ  c$ for all $a,b,c\in G$.
\end{defn} 
 The additive identity of a brace $A$ will be denoted by $0$ and the multiplicative identity by $1$. In every brace $0=1$. 
The same notation will be used for skew braces (in every skew brace $0=1$).

Some authors use the notation $\cdot $  instead of $\circ $ and $*$ instead of $\cdot $ (see for example \cite{[6], gateva, Gateva}).

In \cite{ESS}, Etingof, Schedler and Soloviev introduced  the retract relation for any  solution $(X,r)$. Denote  $X=\{x_{1}, \ldots , x_{\mathcal N}\}$ and
$r(x,y)=(\sigma _{x}(y), \tau _{y}(x))$.  
Recall that the retract relation $\sim $ on $X$  is defined by $x_{i}\sim   x_{j}$  if
$\sigma _{x_i} = \sigma _{x_j}$. The  induced solution
$Ret(X, r) = (X/\sim , r^{\sim })$  is called the {\em retraction } of $X$. A solution $(X, r)$ is
called  a {\em  multi-permutation solution } of level $m$ if $m$ is the smallest non-negative
integer such that  after $m$ retractions we obtain the solution with one element.

$ $

Throughout this paper we will use the following result,  which  is  implicit in \cite{[25], [26]} 
and explicit in Theorem 4.4 of \cite{[6]}.

\begin{thm}\label{Rump}(Rump's theorem, \cite{[25], [26], [6]}).  
Assume  $(B, +, \circ)$ is a brace. If the map  $\check r_B: B\times B \to B \times B$ is defined as 
${\check r}_B(x,y)=(\sigma _{x}(y), \tau _{y}(x))$, where $\sigma _{x}(y)=x\circ y-x$, $\tau _{y}(x)=t\circ x-t$, and $t $ is the inverse of $\sigma _{x}(y)$ in the circle group $(B, \circ ),$ then  $(B, \check r_B)$ is an involutive,  non-degenerate solution of the braid equation.\\
Conversely,  if $(X,\check r)$ is an involutive, non-degenerate solution of the braid equation, then there exists a brace $(B,+, \circ)$ (called an  underlying brace of the solution $(X, \check r)$) such that $B$ contains $X,$ $\check r_B(X\times X )\subseteq X \times X$, and the map $\check r$ is equal to the restriction of $\check r_B$ to $X \times X.$ Moreover, both the additive $(B, +)$ and multiplicative $(B,\circ)$ groups of the brace $(B,+, \circ)$ are generated by $X.$
\end{thm}

 We will call the brace $B$  an underlying brace of the solution $(X,{\check r})$, or a brace associated to the  
solution $(X,{\check r})$. We will also say that the solution $(X, \check r )$ is associated to brace $B$. Notice that this is also related  to the formula of
 set-theoretic solutions associated to the braided group (see \cite{ESS} and \cite{gateva}). 

The following fact was also discovered by Rump.
 
\begin{rem} \label{nilpotent}
 Let $(N, +, \cdot)$ be an associative  nilpotent ring.  If for $a,b\in N$ we define 
 \[a\circ b=a\cdot b+a+b,\] then $(N, +, \circ )$ is a brace.
 \end{rem}

\begin{defn}
 %Let $X, Y$ be sets and ${\check r}:X\times X\rightarrow X\times X$, ${\check r}':Y\times %Y\rightarrow Y\times Y$ be functions.
Let $(X, {\check r})$ and $(Y,{\check r}')$ be set-theoretic  solutions of the braid equation, and let $f:X\rightarrow Y$ 
be a function such that 
${\check r}'(f(x), f(y))=(f\times f)({\check r}(x,y))$,  for all $x,y\in X$. Then $f$ is called a homomorphism of solutions.
If $f$ is $one-to-one$ then $f$ is called an isomorphism of solutions.
 \end{defn}

\begin{lemma}
Notice that if a solution $(X, {\check r})$ comes from brace $B$, and $J$ is an ideal in $B$ and $X_{J}=\{x+J: x\in X\}$ is  
a subset of the factor brace $B/J$  then the map  
 $f:X\rightarrow X_{J}$ given by $ f(x)=x+J$ is a homomorphism of solutions $(X,{\check r})$ and $(X_{J}, {\check r}_{J})$,
 where $(X_{J}, {\check r}_{J})$
 is the solution associated to the brace $B/J$ on the set $X_{J}$. 
\end{lemma}
\begin{proof} 
It  follows immediately from the properties of an ideal in a brace (ideals in braces were defined  in \cite{[26]}. See also  \cite{[6]}).
\end{proof}

%\subsection{Basic information about $\mathfrak {gl}_{\mathcal N}$ and the correponding Yangian} 

\section{Hecke algebras $\&$ quantum groups from braces}

\subsection{The Yang-Baxter equation $\&$  Hecke algebras}

\noindent In this section we explore various connections between braces, representations of the $A$-type Hecke algebras, 
and quantum algebras. In particular, after showing some 
fundamental properties for the brace $R$-matrices and making the direct connection with $A$-type Hecke algebras, 
we derive new quantum algebras coming from braces. The Yangian ${\mathcal Y}(\mathfrak{gl}_{\mathcal N})$ turns out to
 be a special case within this larger class of quantum algebras.

Before we start our investigation on the aforementioned connections
let us first derive some preliminary results, that will be essential especially when proving the integrability 
of open spin-chain like systems, this issue however is discussed separately in \cite{DoiSmo}.
Recall the Yang-Baxter equation in the braid form  in the presence of spectral parameters 
$\lambda_1,\ \lambda_2$ ($\delta = \lambda_1 - \lambda_2$):
\begin{equation}
\check R_{12}(\delta)\ \check R_{23}(\lambda_1)\ \check R_{12}(\lambda_2) = \check R_{23}(\lambda_2)\
 \check R_{12}(\lambda_1)\ \check R_{23}(\delta) . \label{YBE1}
\end{equation}
where $\check R: V \to V,$  and let in general $\check R = \sum_{j} a_j \otimes b_j,$ then in the index notation 
$\check R_{12} =\sum_j a_j \otimes b_j \otimes I_V,$  $\check R_{23} =\sum_j  I_V \otimes a_j \otimes b_j $ and  $\check R_{13} =\sum_j a_j \otimes  I_V \otimes b_j.$

We focus here on Baxterized solutions of  (\ref{YBE1}) coming form braces, i.e.
\begin{equation}
\check R(\lambda) = \lambda \check r + {\mathbb I}, \label{braid1}
\end{equation}
where ${\mathbb I}= I_X\otimes I_X$  and $I_{X}$ is the identity matrix of dimension equal to the cardinality of the set $X$.
Indeed, (\ref{braid1}) satisfies (\ref{YBE1}), provided that $\check r$ satisfies the braid equation and $\check r^2 = I_{X\otimes X}.$
Also, we recall  the notation introduced in Lemma 2.1 for the matrix $\check r$ (\ref{brace1}). Let also, $R = {\mathcal P} \check R$, then
\begin{equation}
R(\lambda)= \lambda r + {\mathcal P}, \label{braid2}
\end{equation}
where $r$ is defined in (\ref{brace2}), and $R$ is a solution of the Yang-Baxter equation in the familiar form:
\begin{equation}
 R_{12}(\delta)\  R_{13}(\lambda_1)\  R_{23}(\lambda_2) = R_{23}(\lambda_2)\ R_{13}(\lambda_1)\ R_{12}(\delta). \label{YBE2}
\end{equation}

\begin{rem}
It would be useful for the following Proposition to introduce the notion of partial transposition. 
Let  $A\in \mbox{End}\big ({\mathbb C}^{\mathcal N} \otimes {\mathbb C}^{\mathcal N} \big )$
expressed as: $A = \sum_{i, j,k,l=1 }^{\mathcal N}A_{ij, kl}\ e_{i,j}\otimes e_{k,l}$.
 We define the {\it partial transposition}  as follows (in the index notation):
\begin{eqnarray}
A_{12}^{t_1} = \sum_{i, j,k,l=1}^{\mathcal N} A_{ij, kl}\ e_{i, j}^t \otimes e_{k, l},\
\quad A_{12}^{t_2} = \sum_{i, j,k,l =1}^{\mathcal N} A_{ij, kl}\ e_{i, j} \otimes e_{k, l}^t
\end{eqnarray}
where $e_{i,j}^t = e_{j,i}$.
\end{rem}

\begin{pro}\label{555}  The brace $R$-matrix satisfies the following fundamental properties:
\begin{eqnarray}
&&  R_{12}(\lambda)\  R_{21}(-\lambda) = (-\lambda^2 +1)  {\mathbb I}, ~~~~~~~~~~~~~\mbox{{\it Unitarity}} \label{u1}\\
&&  R_{12}^{t_1}(\lambda)\ R_{12}^{t_2}(-\lambda -{\mathcal N}) = \lambda(-\lambda -{\mathcal N}) {\mathbb I}, ~~~~~
\mbox{{\it Crossing-unitarity}} \label{u2}\\
&& R_{12}^{t_1 t_2}(\lambda) = R_{21}(\lambda), \label{tt}
\end{eqnarray}
{\it where $^{t_{1,2}}$ denotes transposition on the first, second space respectively.}
\end{pro}
\begin{proof}
Recall $R_{21} = {\mathcal P} R_{12} {\mathcal P}$, the proof of unitarity is straightforward due to $\check r^2 = {\mathcal P}^2 = {\mathbb I}$. 
To prove crossing-unitarity (\ref{u2})  it suffices to show the following identities:
\begin{equation}
({\mathcal P}_{12}^{t_1})^2 = {\mathcal N} {\mathcal P}_{12}^{t_1},\  \quad  r_{12}^{t_1} {\mathcal P}_{12}^{t_1} = 
{\mathcal P}_{12}^{t_1} r_{12}^{t_2} = {\mathcal P}_{12}^{t_1},\  \quad r_{12}^{t_1} r_{12}^{t_2} =  {\mathbb I}. \label{p3}
\end{equation}
The above can be easily shown, from the definitions of ${\mathcal P} = \sum_{x,y} e_{x,y} \otimes e_{y,x}$ and $r= {\mathcal P}\check r$ (\ref{brace2}).
%\begin{eqnarray}
%{\mathcal P}= \sum_{x,y} e_{x,y} \otimes e_{y,x},\  \quad  r =\sum_{x,y} e_{y, \sigma_x(y) } \otimes e_{x,\tau_y(x)} . \label{w3}
%\end{eqnarray}
Given (\ref{p3}) the crossing-unitarity immediately follows. 
The last property (\ref{tt}) immediately follows from the definitions of $R_{12},\  R_{21}$ and the brace representation.
\end{proof}

We can now state the obvious connection of the brace representation with the $A$-type Hecke algebra.

\begin{defn} The $A$-type Hecke algebra ${\mathcal H}_N(q)$ is defined by 
the generators $g_l$, $l \in \{1,\ 2, \ldots, N-1 \}$ and the exchange relations:
\begin{eqnarray}
&& g_l\ g_{l+1} \ g_l = g_{l+1}\ g_l\ g_{l+1}, \label{h1}\\
&& \Big [ g_l,\ g_m \Big ]= 0, ~~~|l-m|>1 \label{h2}\\
&& \big ( g_l-q \big )\big (g_l+q^{-1} \big) =0. \label{h3}
\end{eqnarray}
\end{defn}

\begin{rem} The brace solution $\check r$ (\ref{brace1})
is a representation of the $A$-type Hecke algebra for $q=1$.\\
Indeed, $\check r$ satisfies the braid relation and $\check  r^2 = {\mathbb I}$ , 
which can be easily shown by using the involution. We can then define $g_{1}={\check r} \otimes I^{\otimes N}$ 
and $g_{i}= I^{\otimes i-1}\otimes  {\check r} \otimes I^{\otimes {N-i-1}}$.
Let us also show below the braid relation:
\begin{eqnarray}
(\check r \otimes I_X)\  ( I_X \otimes \check r)\ (\check r \otimes I_X ) =  (  I_X \otimes \check r)\ 
(\check r \otimes I _X)\ (\check r \otimes I_X ).
\end{eqnarray}
The LHS of the equation above is equal to
\begin{equation}
\sum e_{x, \sigma_{\bar x}(\bar y)} \otimes e_{y, \tau_{\bar y}(\bar x)}\otimes e_{\hat y, \tau_{\hat y}(\hat x)}
\end{equation}
provided that $\hat x = \tau_{y}(x),\ \bar x = \sigma_{x}(y),\ \bar y = \sigma_{\hat x}(\hat y)$.\
 We can change $\hat y$ to be denoted by $z$, to obtain:
\begin{equation}
\sum e_{x, \sigma_{\bar x}(\bar y)} \otimes e_{y, \tau_{\bar y}(\bar x)}\otimes e_{z, \tau_{z}(\hat x)}
\end{equation} 
provided that $\hat x = \tau_{y}(x),\ \bar x = \sigma_{x}(y),\ \bar y = \sigma_{\hat x}(z)$.\

 On the other hand, denote $r_{1}(x,y,z)=(\sigma_{x}(y), \tau_{y}(x), z)$, $r_{2}(x,y,z)=(x, \sigma_{y}(z), \tau _{z}(y)).$
 Notice that,  using the same notation as in the expresion for the LHS,  we obtain 
\[r_{1}r_{2}r_{1}(x,y,z)=(\sigma_{\bar x}(\bar y), \tau_{\bar y}(\bar x), \tau_{z}(\hat x)).\]

Similarly, the RHS is equal to
\begin{equation}
\sum e_{x,\sigma_{x}(y)} \otimes e_{ x'', \sigma_{ x'}( y')}\otimes e_{ y'', \tau_{ y'}( x')}
\end{equation}
provided that $y= \sigma_{ x''}(y''),\  x' = \tau_{y}(x),\  y' = \tau_{ y''}(x'')$.\\
 We can denote the variable $y$ by $t$ and then change  the variable $x''$ to $y$, and variable $y''$ to $z$ we obtain that the RHS is equal to 
\begin{equation}
\sum e_{x,\sigma_{x}(t)} \otimes e_{ y, \sigma_{ x'}( y')}\otimes e_{ z, \tau_{ y'}( x')}
\end{equation}
provided that $t= \sigma_{ y}(z),\  x' = \tau_{t}(x),\  y' = \tau_{ z}(y)$.\\

Observe now that using the same notation as in the RHS we obtain  
\[r_{2}r_{1}r_{2}(x,y,z)=(\sigma_{x}(t), \sigma_{ x'}( y'),   \tau_{ y'}( x')).\]

Equality between the LHS and RHS expressions is guaranteed for the brace solution and follows from Rump's theorem (Theorem \ref{Rump}), 
since $r_{1}r_{2}r_{1}(x,y,z)=r_{2}r_{1}r_{2}(x,y,z)$ by Rump's theorem (because the map $(x,y)\rightarrow (\sigma _{x}(y), \tau_{y}(x))$ satisfies the set-theoretic solution of the 
 Braid equation). 
\end{rem}

\subsection{Quantum algebras from braces}

\noindent 
Given a solution of the Yang-Baxter equation, an associated quantum algebra can be identified, 
within the so-called  Faddeev-Reshetikhin-Takhtajan (FRT) construction \cite{FadTakRes},
via the fundamental relation (we have multiplied the familiar RTT relation by the permutation operator):
\begin{equation}
\check R_{12}(\lambda_1 -\lambda_2)\ L_1(\lambda_1)\ L_2(\lambda_2) = L_1(\lambda_2)\ L_2(\lambda_1)\ \check 
R_{12}(\lambda_1 -\lambda_2),\label{RTT}
\end{equation} 
where $\check R(\lambda) \in \mbox{End}({\mathbb C}^{{\mathcal N}} \otimes {\mathbb C}^{{\mathcal N}})$, $\ L(\lambda) \in 
\mbox{End}({\mathbb C}^{{\mathcal N}}) \otimes {\mathfrak A}$, and ${\mathfrak A}$ is the quantum algebra 
defined by (\ref{RTT}). Here we  have used  the``index notation'', i.e.  we define\footnote{Notice that in $L$ 
in addition to the indices 1 and 2 in (\ref{RTT}) there is also an implicit ``quantum index'' $n$ associated to ${\mathfrak A},$ 
which for now is omitted, i.e. one writes $L_{1n},\ L_{2n}$.} 
\begin{eqnarray}
&&  L_1(\lambda) = \sum_{z, w = 1}^{\mathcal N} e_{z,w} \otimes I_{\mathcal N}\otimes L_{z,w}(\lambda) \label{ll1} \\
&&  L_2(\lambda)= \sum_{z, w =1}^{\mathcal N} I_{\mathcal N} \otimes  e_{z,w}  \otimes L_{z,w}(\lambda)  \label{ll2}  \\
&& \check R_{12} = \check R \otimes 1_{\mathfrak A }\label{def}
\end{eqnarray}
 where $I_{\mathcal N}$ is the ${\mathcal N}$ dimensional identity matrix,  $1_{\mathfrak A }$ is the unit element of the algebra ${\mathfrak A }$ and
$L_{z,w}(\lambda)$ are elements of the affine algebra ${\mathfrak A}$ (defined by (\ref{RTT})). The quantum algebra is equipped with a co-product $\Delta:{\mathfrak A} 
\to {\mathfrak A} \otimes {\mathfrak A}$ \cite{FadTakRes, Drinfeld}. Indeed, we define 
${\mathrm T}_{1,23}(\lambda)= L_{13}(\lambda) L_{12}(\lambda),$ which satisfies (\ref{RTT})
 and is expressed as ${\mathrm T}_{1,23}(\lambda)  = \sum_{x,y \in X} e_{x,y} \otimes \Delta(L_{x,y}(\lambda)).$

We shall focus now on solutions of the Yang-Baxter equation coming from braces and will identify the defining relations 
of the associated quantum algebra  via (\ref{RTT}).

\begin{pro}\label{Q}  The quantum algebra associated to the solution $\check R =I_{X\otimes X} + 
 \lambda  \check r,$ where $\check r$ is the  brace solution 
\[{\check r}=\sum_{x, y\in X} e_{x, \sigma _{x}(y)}\otimes e_{y, \tau _{y}(x)}\] 
is defined by generators $L^{(m)}_{z,w},\ z, w \in X$, and defining relations 
\begin{eqnarray}
L_{z,w}^{(n)} L_{\hat z, \hat w}^{(m)} - L_{z,w}^{(m)} L_{\hat z, \hat w}^{(n)} &=& 
L^{(m)}_{z,\sigma_w(\hat w)} L^{(n+1)}_{\hat z, \tau_{\hat w}( w)}- L^{(m+1)}_{z, \sigma_w(\hat w)} 
 L^{(n)}_{\hat z, \tau_{\hat w}( w)}\nonumber\\ &-& L^{(n+1)}_{ \sigma_z(\hat z),w} 
 L^{(m)}_{\tau_{\hat z}( z), \hat w }+ L^{(n)}_{ \sigma_z(\hat z),w}  L^{(m+1)}_{\tau_{\hat z}( z), \hat w}. \label{fun2}
\end{eqnarray}
 \end{pro}
\begin{proof} 
We  express $L$ as a formal power series expansion $L_{a}(\lambda) = \sum_{n=0}^{\infty} {L_{a}^{(n)} \over \lambda^n}$, 
$a \in \{1, 2 \}.$ 
Substituting  expression (\ref{braid1}), and the $\lambda^{-1}$ expansion of $L_a(\lambda)$ in (\ref{RTT}) we obtain
the defining relations of the quantum algebra associated 
to a brace $R$-matrix (we focus on terms $\lambda_1^{-n} \lambda_2^{-m}$):
\begin{eqnarray}
&&  \check r_{12} L_{1}^{(n+1)} L_2^{(m)} -\check  r_{12} L_1^{(n)} L_2^{(m+1)} +  L_1^{(n)} L_2^{(m)} \nonumber\\
&&  = L_1^{(m)} L_{2}^{(n+1)} \check r_{12} -  L_1^{(m+1)} L_2^{(n)}\check r_{12} +  L_1^{(m)} L_2^{(n)}. \label{fund}
\end{eqnarray}
Equation (\ref{fund}) immediately leads to the quantum algebra relations (\ref{fun2}), after recalling  similarly to (\ref{ll1})-(\ref{def}):
\begin{eqnarray}
L_{1}^{(k)}=\sum_{i,j\in X}e_{i,j}\otimes I_X \otimes  L_{i,j}^{(k)}, \quad  L_{2}^{(k)}=\sum_{i,j\in X} I_X \otimes 
e_{i,j}\otimes  L^{(k)}_{i,j}, \nonumber
\end{eqnarray}
$\check r_{12} = \check r \otimes  1_{\mathfrak A }$, $I_X$ is the identity matrix of dimension equal to the cardinality of the set $X$, 
and $L^{(k)}_{i,j} $ are the generators of the associated quantum algebra.

%Recall that the relations for the quantum algebra are:

%\[{\check r}_{1,2}L_{1}^{(n+1)}L_{2}^{(m)}-{\check r_{1,2}}L_{1}^{(n)}L_{2}^{(m+1)}+L_{1}^{(n)}L_{2}^{(m)}=\]
%\[=L_{1}^{(m)}L_{2}^{(n+1}{\check r}_{1,2}L_{1}^{(m+1)}L_{2}^{(n)}{\check r}_{1,2}+L_{1}^{(m)}L_{2}^{(n)}.\]

Indeed, by substituting the above expressions for $L_{1}^{(k)}$ and $L_{2}^{(k)}$ and 
${\check r}=\sum_{i,j\in X}e_{i, { }^{i_j}}\otimes e_{j, i^{j}} $ in (\ref{fund}) and 
computing  both sides we obtain: 

\[\sum_{x, j, y, i\in X} e_{x, j}\otimes e_{y,i}\otimes Q_{x, j, y, i}^{(m,n)}=\sum_{x, j, y, i\in X} e_{x, j}
\otimes e_{y, i}\otimes P_{x, j, y, i}^{(m,n)}\] where 

\[Q_{x, j, y, i}^{(m,n)}=   L_{{{ }^xy}, j}^{(n+1 )}L_{{x^{y}},i}^{(m )}- L_{{{ }^xy}, j}^{(n )}L_{{x^{y}},i}^{(m+1)}
 + L_{x, j}^{(n)}L_{y,i}^{(m)},\]

\[P_{x, j, y, i}^{(m,n)}=
L_{x, {{ }^ji}}^{(m )}L_{y, j^{i}}^{(n+1 )}-
L_{x, {{ }^ji}}^{(m+1 )}L_{y, j^{i}}^{(n )}+
 L_{x,j}^{(m)}L_{y,i}^{(n)}.\]

 Notice that 
$Q_{x, j, y, l}^{(m,n)}=P_{x, j, y, l}^{(m,n)}$ are the defining relations in our quantum algebra.

\end{proof}

\begin{defn} Let $(X, {\check r})$ be a set-theoretic solution of the Yang-Baxter equation, 
 with ${\check r}(x,y)=({{ }^xy},x^{y})$. 
 The quantum algebra associated to the brace $\check R$ matrix  $\check R(\lambda)=I_{X\otimes X}+\lambda {\check r}$
is defined by generators $L^{(m)}_{z,w},\ z, w \in X$, $m=0, 1, 2, \ldots $ and defining relations 
 \[L_{{{ }^xy}, j}^{(n+1 )}L_{{x^{y}},i}^{(m )}- L_{{{ }^xy}, j}^{(n )}L_{{x^{y}},i}^{(m+1 )}
 + L_{x, j}^{(n)}L_{y,i}^{(m)}=\]
\[L_{x, {{ }^ji}}^{(m )}L_{y, j^{i}}^{(n+1 )}-
L_{x, {{ }^ji}}^{(m+1 )}L_{y, j^{i}}^{(n )}+
 L_{x,j}^{(m)}L_{y,i}^{(n)}\]
 for $x,j,y,i\in X$. This algebra will be denoted as ${\mathfrak A}{(X, {\check r})}$.
\end{defn}

This is the same algebra as in Proposition \ref{Q}. 

%\noindent {\bf Special case: The Yangian ${\mathcal Y}(\mathfrak {gl}_{\mathcal N})$.}

 In this part we recall some basic notions about the $\mathfrak{gl}_{ \mathcal N}$ algebra
 and the Yangian  ${\mathcal Y}(\mathfrak {gl}_{\mathcal N})$.

\begin{defn} {\label{defg}}
The $\mathfrak {gl}_{\mathcal N}$ algebra is a Lie algebra (over $\mathbb C$) with generators denoted as ${\mathfrak l}_{i,j}$ 
for $i,j\in \{1, \ldots , {\mathcal N}\}$, 
that satisfy:
  \begin{equation}
\Big [{\mathfrak l}_{i,j},\ {\mathfrak l}_{k,l}\Big ]={\mathfrak l}_{i,l}\delta _{k,j}, -{\mathfrak l}_{k,j}\delta _{i,l}.\label{gln1}
\end{equation}
\end{defn}

\begin{rem}{\label{remg}}  Recall that $e_{x, y}$ are ${\mathcal N} \times {\mathcal N}$ matrices with elements 
$(e_{x,y})_{z,w}=\delta_{x,z} \delta_{y,w},$  $x, y \in \{1,...,{\mathcal N} \}.$
We call  the ${\mathcal N}$ dimensional representation $\rho: \mathfrak{gl}_{\mathcal N} \to \mbox{End}({\mathbb C}^{\mathcal N}),$  such that 
${\mathfrak l}_{x,y} \mapsto e_{x,y}$ the fundamental representation of $\mathfrak{gl}_{\mathcal N},$
i.e. the matrices $e_{x,y}$  
are the generators of $\mathfrak{gl}_{\mathcal  N}$  in the fundamental (${\mathcal N}$-dimensional) representation.
\end{rem}
The $\mathfrak {gl}_{\mathcal N}$ algebra is equipped with a coproduct 
$\Delta: \mathfrak{gl}_{\mathcal N} \to  \mathfrak{gl}_{\mathcal N} \otimes \mathfrak{gl}_{\mathcal N},$ such that
\begin{equation}
\Delta({\mathfrak l}_{i,j})= 1 \otimes {\mathfrak l}_{i,j}+ {\mathfrak l}_{i,j} \otimes 1,  \quad  \forall\  {\mathfrak l}_{i,j} \in  \mathfrak{gl}_{\mathcal N}.
\end{equation}
The $N$-coproduct is obtained by iteration $\Delta^{(N)}= (\mbox{id} \otimes \Delta^{(N-1)})\Delta = (\Delta^{(N-1)} \otimes \mbox{id})\Delta$ (co-associativity holds):
\begin{equation}
\Delta^{(N)}({\mathfrak l}_{i,j}):= \sum_{n=1}^N ({\mathfrak l}_{i,j})_{n}  = \sum_{n=1}^N 1 \otimes \ldots 
\otimes\underbrace{{\mathfrak l}_{i,j}}_{n^{th}\  \mbox{position}} \otimes \ldots \otimes 1.
 \label{coproduct1}
\end{equation}
The element $({\mathfrak l}_{i,j})_n$, in the standard {\it index notation} above, appears in the $n^{th}$ position of the $N$ co-product,
also $1$ is the unit element of the algebra and we set $\Delta^{(2)}({\mathrm Y}) := \Delta({\mathrm Y})$.
The elements $\Delta^{(N)}(e_{i,j})$  are tensor representations of $\mathfrak{gl}_{\mathcal N}$, i.e.
$\mathfrak {gl}_{\mathcal N} \to \mbox{End}(({\mathbb C}^{\mathcal N})^{\otimes N})$, such that ${\mathfrak l}_{i,j} \mapsto \Delta^{(N)}(e_{i,j})$.

We recall two well known results about Yangians. For more information on Yangians we refer the interested reader to 
\cite{Drinfeld, yangians}. 
The Yangian for  $\mathfrak {gl}_{\mathcal N}$ will be denoted as  ${\mathcal Y}(\mathfrak {gl}_{\mathcal N})$.

\begin{defn} The Yangian  ${\mathcal Y}(\mathfrak {gl}_{\mathcal N})$ is an algebra over ${\mathbb C}$ (an infinite extension of 
$\mathfrak{gl}_{\mathcal N}$)
with generators denoted as $L_{i,j}^{(n)},$ for $i,j \in \{1,2, \ldots, {\mathcal N}\},$ $n \in {\mathbb N},$ ($L_{i,j}^{(0)} = \delta_{i,j}$) that satisfy the defining relations:
\begin{equation}
\Big [ L_{i,j}^{(n+1)},\ L_{k,l}^{(m)}\Big ] -\Big [ L_{i,j}^{(n)},\ L_{k,l}^{(m+1)}\Big ] = L_{k,j}^{(m)}L_{i,l}^{(n)}- L_{k,j}^{(n)}L_{i,l}^{(m)}.
 \label{fund2b}
\end{equation}
\end{defn}
The latter relations are defined up to an overall multiplicative constant that can be absorbed by rescaling the generators.

\begin{cor}\label{gln}  In the special case $\check r ={\mathcal P}$ the ${\mathcal Y}(\mathfrak {gl}_{\mathcal N})$ algebra is recovered 
as the corresponding  quantum algebra from (\ref{fund}).
\end{cor}
\begin{proof} 
We consider the special case where $\check r={\mathcal P}$ in (\ref{fund}), which corresponds to the Yangian ${\mathcal Y}(\mathfrak{ gl}_{\mathcal N})$.  Recall that we express $L$ as a formal power series expansion $L(\lambda) = I_X \otimes 1 + \sum_{n=1}^{\infty} {L^{(n)} \over \lambda^n}$ ($L^{(0)} = I_X \otimes 1$).
Then the fundamental relation (\ref{fund})  leads to:
\begin{equation}
\Big [ L_1^{(n+1)},\ L_2^{(m)}\Big ] -\Big [ L_1^{(n)},\ L_2^{(m+1)}\Big ]  = {\mathcal P}_{12} \Big ( L_1^{(m)}L_2^{(n)}-L_1^{(n)}L_2^{(m)}\Big ).
\label{fund2}
\end{equation}

Recalling that $L_1^{(n)} = \sum_{x,y \in X }e_{x,y} \otimes I \otimes L_{x,y}^{(n)},$  
$L_2^{(n)} = \sum_{x,y \in X } I \otimes e_{x,y}  \otimes L_{x,y}^{(n)},$ and ${\mathcal P}_{12}= \sum e_{ij} \otimes e_{ji}\otimes 1$,  we arrive at (\ref{fund2b}) (see also \cite{Drinfeld}).
\end{proof}

\begin{cor} In the case of the Yangian ${\mathcal Y}(\mathfrak {gl}_{\mathcal N})$ the finite dimensional subalgebra  $\mathfrak {gl}_{\mathcal N}$  emerges from (\ref{fund2}), and is
realized by the elements of $L^{(1)}.$
\end{cor}
\begin{proof}
Considering  (\ref{fund2}) for $n=0$ and $m=1$ we obtain
\begin{eqnarray}
\Big [ L_1^{(1)},\ L_2^{(1)}\Big ] = {\mathcal P}_{12}\Big ( L_1^{(1)} -  L_2^{(1)}\Big )  . \label{A2}
\end{eqnarray}
From (\ref{A2}) we deduce  that the elements $L_{i,j}^{(1)}$ satisfy the defining relations of $\mathfrak {gl}_{\mathcal N}.$  
\end{proof}

Note that the choice $L^{(0)} = I_X \otimes 1$ in the case of the Yangian is compatible with the fact that 
$L(\lambda)=  \lambda I_X \otimes 1+ {\mathfrak L}$, where ${\mathfrak L}= \sum_{i,j=1}^{{\mathcal N}} e_{i,j} 
\otimes{\mathfrak l}_{i, j}$  and ${\mathfrak l}_{i, j}$ are the generators of $\mathfrak{gl}_{\mathcal N}$ (Definition \ref{defg}), 
provides a realization of  ${\mathcal Y}(\mathfrak{ gl}_{\mathcal N})$, \cite{FadTakRes, yangians} (see also next section, comments in the proof of 
Corollary \ref{Corollary1} on tensor representations of the quantum algebra in the special case of Yangian).

After the brief ``interlude''  regarding the Yangian case we return to quantum algebras associated to general  brace solutions.

\begin{pro} \label{useful}
Let $(X, {\check r})$ and $(Y,{\check r}')$ be set-theoretic  solutions of the Braid equation.
If $f:X\rightarrow Y$ is a surjective homomorphism of solutions $(X,{\check r})$ and $(Y, {\check r}')$, then the map ${\mathfrak A}{(X,{\check r})} \to {\mathfrak A}{(Y,{\check r}')},$ determined by
\[ L_{x,y}^{(k)}\mapsto L_{f(x),f(y)}^{(k)}\] is a homomorphism of algebras. 
\end{pro}
\begin{proof} We can verify that this function maps the defining relations of the  quantum algebra ${\mathfrak A}{(X,{\check r})}$ 
onto defining relations of ${\mathfrak A}{(Y,{\check r}')}$.
\end{proof}

By a representation of an algebra  $A$ we will mean a factor algebra $A/I$ where $I$ is an ideal in the algebra $A$.

\begin{pro}\label{factoring}
Let $B$ be a brace, $X$ be a subset of $B$ and $(X,{\check r})$ be an involutive solution of the Braid relation obtained from this brace as in  Lemma 2.7.  Let $J$ be an ideal of the brace $B$. Let $(X_{J}, {\check r}_{J})$ be the solution associated to  the brace $B/J$ on the subset $X_{J}$ of 
 the factor brace $B/J$. Then the algebra ${\mathfrak A}{ (X_{J},{\check r}_{J})}$ is a representation of the algebra 
${\mathfrak A}{ (X,{\check r})}$. 
\end{pro}
\begin{proof} Let $A$ be the free algebra with generators $u_{i,j}^{(n)}$, then we can define  homomorphism of algebras $f: A\rightarrow {\mathfrak A}{(X,{\check r})}$
 by $f(u_{i,j}^{(n)})=L_{i,j}^{(n)}$ and let $I$ be the kernel of this map. 
  By the First Isomorphism Theorem $A/I$ is isomorphic as algebra to ${\mathfrak A}{(X,{\check r})}$. 
  We can also define homomorphism of algebras $g: A\rightarrow {\mathfrak A}{ (X_{J},{\check r}_{J})}$  by $g(u_{i,j}^{(n)})=L_{i+J,j+J}^{(n)}$ and let $T$ be the kernel of $g$. Observe that $I\subseteq T$. 
  By the First Isomorphism Theorem $A/T$ is isomorphic as algebra to ${\mathfrak A}{ (X_{J},{\check r}_{J})}$.
 By the Second Isomorphism Theorem for rings the factor algebra $(A/I)/(T/I)$ is isomorphic to $A/T$, therefore   
${\mathfrak A}{ (X_{J},{\check r}_{J})}$ is a representation 
${\mathfrak A}{ (X,{\check r})}$. 
\end{proof}

We call a solution $(Y, {\check r})$ of the Braid equation  trivial if and only if ${\check r}(x,y)=(y,x)$ for all $x,y\in Y$. In this case we may 
denote ${\check r}$ as $\tau $.

\begin{pro}\label{fact2}  A set-theoretic solution  $(X, {\check r})$  which can be homomorphically mapped onto a trivial solution $(Y, \tau )$ will have 
    ${\mathcal Y}(\mathfrak {gl}_{\mathcal N})$ as  a representation of its algebra ${\mathfrak A}{(X,{\check r})}$, where $\mathcal N$ is the cardinality of $Y$.
\end{pro}
\begin{proof} The trivial solution $(Y, \tau )$ of cardinality ${\mathcal N}$ has  ${\mathcal Y}(\mathfrak {gl}_{\mathcal N})$ as its quantum group by 
Corollary \ref{gln}. By Proposition \ref{factoring},
 a solution $(X, {\check r})$  which can be homomorphically mapped onto a trivial solution $(Y, \tau )$ will have 
    ${\mathcal Y}(\mathfrak {gl}_{\mathcal N})$ as a representation of  its  algebra ${\mathfrak A}{(X,{\check r})}$.
\end{proof}

 Recall that if $(X,{\check r})$ is a  non-degenerate, 
involutive set-theoretic solution of the Braid equation,   and $Y\subseteq X$, $e\in Y$, then $Y$ is an orbit of $e$ if for $x\in X, y\in Y$ we have 
$\sigma _{x}(y)\in Y$ and $\tau _{x}(y)\in Y$ and $Y$ is the smallest set with this property.
\begin{cor}
If $(X,{\check r})$ is an involutive, non-degenerate solution with  $\mathcal N$ orbits, then  ${\mathcal Y}(\mathfrak {gl}_{\mathcal N})$ 
 is a representation of the  algebra ${\mathfrak A}{(X, {\check r})}$.
\end{cor}
\begin{proof}
 We can map $(X,{\check r})$ onto a trivial solution by mapping each element on its orbit.
 It is easy to check that this map is a 
homomorphism of set-theoretic  solutions. 
The result now follows from Proposition \ref{fact2}. 
\end{proof}

 Some basic information about orbits and examples of orbits can be found  in  \cite{LAA}, page 90 (just above section 2.2).
 Solutions which have only one orbit are called indecomposable solutions. Indecomposable involutive non-degenerate set-theoretic 
solutions of the Yang-Baxter equation can be constructed using one-generator braces \cite{LAA}.

Let $m \in {\mathbb Z}^+$ and $X_{m}=\{1, 2, \ldots , m\}$, define also ${\check r}_{m}(i,j)=(j+1, i-1)$ where addition and subtraction are taken modulo $m$. 
This is a special type of Lyubashenko's solution \cite{Drin}. We define ${\mathfrak A}_{m}$ to be the algebra ${\mathfrak A}{(X_{m}, {\check r}_{m})}$.
The following shows that Lyubashenko's solutions are useful for constructing  representations of  algebras constructed from braces. 
\begin{pro}\label{lyub}
Let $(X,{\check r})$ be a finite,  indecomposable, involutive and  non-degenerate set theoretic solution of a finite 
multipermutation level. Suppose that $\sigma _{x}(z)\neq \sigma _{y}(z)$ for some $x,y,z\in X$. 
Then the algebra ${\mathfrak A}{(X, {\check r})}$ associated to the solution $(X, {\check r})$ can be mapped onto 
the algebra ${\mathfrak A}_{m}$ for some $m>1$.
% Moreover, $m $  is a a divisor of the order of the 
%permutation group $\mathcal {G}(X, {\check r})$. 
\end{pro}
\begin{proof} Let $(X_{ret}, {\check r}_{ret})$ be the retraction of $(X,{\check r})$.
 Notice that the retraction of $(X, \check {r})$ has more than one element, since  $\sigma _{x}(z)\neq \sigma _{y}(z)$ for some $x,y,z\in X$.
 Notice that $(X_{ret}, {\check r}_{ret})$ satisfies the assumptions of Proposition  7.1  from \cite{LAA},
 hence it can be mapped onto  solution $(X_{m }, {\check r}_{m })$ for some $m >1$. 
\end{proof}
 %For every $x\in X$ there exists $i$ such that $\sigma _{x}^{i}=Id $ on $X$, where $i$ divides the order of a permutation group 
%of $(X, {\check  r})$. Notice that for $ {\check r}_{m}^{i}(1,y)$ equals either $(y+i, *)$ or $(*, y+, therefore
% $\sigma _{1}^{m}=0$ and $\sigma _{1}^{i}\neq Id $ for $i<m $, therefore $m $ divides 
%the cardinality of the permutation group of the solution  $(X, {\check r})$.

Proposition \ref{lyub} suggests the following question:

{\bf Question 1.} Investigate representations of the algebra ${\mathfrak A}_{m}$ which is associated to  
Lyubashenko's solution.

$ $

\subsection{Representations of algebras} 

We start this subsection with the following observation.
\begin{pro}
Let $(X,{\check r})$ be a set theoretic solution of the braid equation and let $A_{(X,{\check r})}$ be an arbitrary associative algebra generated 
by elements from the set $X$ and satisfying relations $xy=uv$ whenever ${\check r}(x,y)=(u,v)$. The following holds:
\begin{enumerate}
\item  $A_{(X,{\check r})}$  is a representation of the algebra ${\mathfrak A}{(X,{\check r})}$ when we map 
$L_{x,y}^{(n)}$ to $x\in A_{(X,{\check r})}$ for every $n$.
 \item $A_{(X,{\check r})}\otimes A_{(X,{\check r})}$ is a  representation of the algebra ${\mathfrak A}{(X,{\check r})}$ 
when we map $L_{x,y}^{(n)}$ to 
$x\otimes y\in A_{(X,{\check r})}\otimes  A_{(X,{\check r})}$ for every $n$.
\item  Let $R$ be an arbitrary commutative associative  algebra over the field $\mathbb C $ generated by elements $c_{1}, c_{2}, \ldots $.
 Then $A_{(X,{\check r})}\otimes R$  is a representation of the algebra ${\mathfrak A}{(X, {\check r})}$ when we map 
$L_{x,y}^{(n)}$ to $x\otimes c_{n}\in A_{(X, {\check r})}\otimes R$ for every $n$. 
% \item Let $R$ be an arbitrary commutative algebra generated by elements $c_{1}, c_{2}, \ldots $. Then $A_{(X, {\check r})}%\otimes A_{(X,{\check r})}\otimes R$ 
%is a  representation of the quantum group ${\mathfrak A}{(X,{\check r})}$ when we map $L_{x,y}^{(n)}$ to 
%$x\otimes y\otimes c_{n}\in A_{(X,{\check r})}\otimes  A_{(X,{\check r})}\otimes R$ for every $n$.
\item In the above point $3$   if $0=c_{2}=c_{3}=\ldots $ then we obtain a representation satisfying 
 $L(\lambda )=L_{0}+\lambda L_{1}$ (where $L_{1}, L_{0}$ are independent of $\lambda $).
\end{enumerate}
\end{pro}
\begin{proof} The proof is by verifying that the algebra relations will go to zero after applying the above homomorphisms of algebras
 (related to the above representations).
\end{proof}

\noindent {\em Examples of algebras $A_{(X,{\check r})}.$} Let $(X, {\check r})$ be an indecomposable, involutive set-theoretic solution of the Yang-Baxter equation. 
Various types of algebras belonging to the class $A_{(X,{\check r})}$ were investigated extensively by several authors, and a lot  is known about them \cite{jk, EJJO}. 
{\em Quantum binomial algebras} were introduced and investigated by Gateva-Ivanova  in 
\cite{15}, \cite{GI3} and \cite{13}. 
{\em The monomial algebras of $I$ type} were investigated for involutive solutions in \cite{20}, \cite{24}, \cite{25}, \cite{26},
 and  recently for both involutive and  
non-involutive solutions  in \cite{Gateva} and \cite{JKA}. 
{\em  The structure algebras of set-theoretic solutions.}  The algebra generated by the set $X$ and with defining relations $xy = uv$ if ${\check r}(x,y) = (u,v)$
was investigated in \cite{JKA}, where in section 5 they study prime ideals in such algebras and hence representations of such algebras which are prime 
(they showed that under mild assumptions they correspond to prime ideals of a group algebra associated to the same set-theoretic solution).  
For involutive solutions such algebras were previously investigated in  \cite{20, 24, 25, 26}. Let $G$ be  a permutation group of a finite, 
non-degenerate, involutive set-theoretic  solution $(X,{\check r})$ of the Braid equation. 
Then $A_{(X,{\check r})}$ can be taken to be {\em the group algebra ${\mathbb C}[G]$}. Such algebras were investigated in \cite{ESS}.\\

\noindent {\em Research directions and open questions.}

\noindent {\em Question 1.} Let $(X,{\check r})$ be an involutive, non-degenerate set-theoretic solution. 
Does the algebra ${\mathfrak A}_{(X,{\check r})}$ have a finite  Gr{\" o}bner basis?

\noindent {\em Question 2.} What can be said about the algebra ${\mathfrak A}_{(X,{\check r})}$ associated to a solution $(X,{\check r})$
 of a finite multipermutation level? Or of an indecomposable solution?   Does  ${\mathfrak A}_{(X, {\check r})}$  
 have any representations of small dimensions?

\section{Novel class of quantum integrable systems $\&$ associated symmetries}
\noindent  In this section we introduce physical spin-chain like systems, with periodic boundary conditions, associated to braces,
and we investigate the corresponding symmetries. Specifically, in the next subsection we provide the general setup for constructing 
integrable quantum spin chains via tensor representations of quantum algebras \cite{FadTakRes}. The transfer matrix, 
which will be defined in subsection 4.1, is 
the generating function of a family of mutually commuting quantities, which guarantee in principle the quantum integrability of the spin-chain system. 
For instance, the momentum and Hamiltonian of the system belong to this family of commuting quantities.
 
In subsection 4.2 we focus on integrable systems constructed from brace solutions of the Yang-Baxter equation
and we investigate the existence of  classes of symmetries of the  corresponding periodic transfer matrix. 
When we say symmetries of the transfer matrix we mean
families of objects ($ {\mathcal N}^{N} \times {\mathcal N}^{N}$ matrices in the cases examined here),
that commute with the transfer matrix and of course they do not belong to the family of mutually commuting quantities. 
Interestingly some of these symmetries 
consist of objects that form certain non-commutative algebras, as will be clear in subsection 4.2.  The properties of the braces play a 
significant role when identifying the new classes of symmetries.
Note that the knowledge  of symmetries provides invaluable information regarding for instance the multiplicities 
of the spectrum of the transfer matrix,
and this in turn has significant implications on the physical behavior of the system at hand.

Before we derive various new 
families of symmetries of the transfer matrix we first prove one of the most important propositions of 
the present investigation.
Specifically, we show that the periodic transfer matrix constructed from Baxterized solutions of the 
$A$-type Hecke algebra ${\mathcal H}_N(q=1)$ 
can be exclusively expressed in terms of the generators of the 
$A$-type Hecke algebras plus some periodic term. This is 
a universal result that holds for any representation of the $A$-type Hecke algebra 
${\mathcal H}_N(q=1)$.

\subsection{Tensor representations of quantum algebras $\&$ integrable systems }

\noindent Given any solution of the Yang-Baxter equation we define the so-called  monodromy matrix 
$T_{0,12...N}(\lambda) \in \mbox{End}\big ({\mathbb C}^{{\mathcal N}} \otimes({\mathbb C}^{{\mathcal N}})^{\otimes N}\big )$, 
which is a tensor representation of the quantum group (\ref{RTT}), \cite{FadTakRes} 
\begin{equation}
T_{0,12...N}(\lambda) =R_{0N}(\lambda) \ldots R_{02}(\lambda)\  R_{01}(\lambda), \label{mono}
\end{equation}
recall $R = {\mathcal P} \check R$, (e.g. in the case of brace solutions $\check R$ is given by (\ref{braid2}), (\ref{brace2})).
We define also the  transfer matrix ${\mathfrak t}_{12...N}(\lambda) = tr_0 \big (T_{0,12...N}(\lambda)\big )  \in  
\mbox{End}\big (({\mathbb C}^{{\mathcal N}})^{\otimes N}\big )$. 
The monodromy matrix $T$ satisfies (\ref{RTT}), and hence one can show that the transfer matrix provides mutually 
commuting quantities \cite{FadTakRes}: 
(${\mathfrak t}(\lambda) =\lambda^N  \sum_{k} {{\mathfrak t}^{(k)} \over \lambda^k}$)
\begin{equation}
\Big [ {\mathfrak t}(\lambda),\ {\mathfrak t}(\mu)\Big ] =0 \ \Rightarrow\  \Big [ {\mathfrak t}^{(k)},\ {\mathfrak t}^{(l)}\Big ] =0.
\label{invo}
\end{equation}
Note that historically the index $0$ is called ``auxiliary'', whereas the indices $1,2,  \ldots, N$ are called ``quantum'',
and they are usually suppressed for simplicity, as is also done in the 
next subsection, i.e. we simply write $T_0(\lambda)$ and ${\mathfrak t}(\lambda)$.  

 The ultimate goal in the context of quantum 
integrable systems, or any quantum system for that matter, is the identification of the eigenvalues 
and eigenvectors of the corresponding Hamiltonian.
In the frame of quantum integrable systems more specifically there exists a set of mutually commuting 
``Hamiltonians'',  guaranteed by the existence of a quantum $R$-matrix that satisfies the Yang-Baxter equation. 
As already discussed this set of mutually commuting objects is generated by the transfer matrix. 
Thus the derivation of the eigenvalues and eigenstates of the transfer matrix is the 
significant problem within quantum intgrability.
This is in general an intricate task and the typical methodology used is the  Bethe ansatz formulation, 
or suitable generalizations, depending on the problem at hand. A detailed study of this problem for transfer
matrices associated to brace solutions will be presented elsewhere.
 
Here we are focusing primarily on the investigation of possible existing new symmetries of the periodic transfer matrix, as any information regarding the symmetries of the transfer matrix provides for instance  valuable  insight on the multiplicities occurring in the spectrum. It will be transparent in what follows that the study 
of the symmetries of the $R$-matrices is a first step towards
formulating the symmetries of the transfer matrix. A detailed analysis on more generic symmetry algebras  and boundary conditions is presented in \cite{DoiSmo}.

\subsection{Symmetries of the periodic transfer matrix of novel classes of spin chains}
\noindent  The main objective in this subsection is the investigation of the symmetries of the periodic transfer matrix for quantum spin chains 
constructed for brace solutions of the Yang-Baxter equation.
We first  present one key proposition, which will have significant implications when studying the symmetries of the period transfer matrix. 
This result  becomes even more prominent when integrable boundary conditions are implemented to integrable  systems, especially those coming from braces \cite{DoiSmo}. 

The following Proposition \ref{Traces} and Lemma 4.2 are quite general and hold 
for any $R(\lambda) =  \lambda {\mathcal P}\check r +{\mathcal P}$, where $\check r$  provides a representation of the $A$-type 
Hecke algebra ${\mathcal H}_N(q=1)$, i.e. satisfies the braid relation and $\check r^2 = {\mathbb I}$, and ${\mathcal P}$ is the permutation operator.
\begin{pro}\label{Traces}  Consider the $\lambda$-series expansion of the monodromy matrix: $T(\lambda) = 
\lambda^N\sum_{k=0}^N{T^{(k)} \over \lambda^k}$ for any $R(\lambda) = \lambda {\mathcal P}\check r +{\mathcal P}$, where $\check r$  
provides a representation of the $A$-type Hecke algebra ${\mathcal H}_N(q=1)$.  Let also  $H^{(k)} = {\mathfrak t}^{(k)} ({\mathfrak t}^{(N)})^{-1}$, $\ k = 0, 
\ldots, N-1$ and $H^{(N)} = {\mathfrak t}^{(N)}$, where 
${\mathfrak t}^{(k)} = tr_0 (T_0^{(k)})$. 
Then the commuting quantities, $H^{(k)}$ for $k = 1, \ldots, N-1$, are expressed exclusively in terms 
of the elements $\check r_{n \, n+1}$, $n = 1, \ldots, N-1$, and $\check r_{N \, 1}$.
\end{pro}
\begin{proof}
%First, from the ${1\over \lambda}$ series expansion of the monodromy matrix $T(\lambda) = \lambda^{N}\sum_{k} {T^{(k)}\over \lambda^k}$,
%we obtain tensor representations of the quantum algebra deirved in Proposition \ref{Q}:
%\begin{eqnarray}
%T_0(\lambda) &=& \lambda^N \Big ({\mathbb R}_{0N1} + {1\over \lambda}  \sum_{n}{\mathbb R}_{0Nn+1} {\mathcal P}_{0n} {\mathbb R}_{0n-11} 
% \nonumber\\ & &+\ {1\over \lambda^2}  \sum_{n_2>n_1}{\mathbb R}_{0Nn_2+1} {\mathcal P}_{0n_2} {\mathbb R}_{0n_2-1n_1+1} {\mathcal P}_{0n_1}
%{\mathbb R}_{0n_1-1 1}+ \ldots \Big ) \label{exp}
%\end{eqnarray}
%where we define ${\mathbb R}_{0nm} = r_{0n}\  r_{0n-1} \ldots r_{0m},\ n>m$.
Let us introduce some useful notation. We define, for $0\leq m< n\leq N+1 $:
${\mathbb P}_{n-1;m+1} = {\mathcal P}_{0 \, n-1}\ldots {\mathcal P}_{0 \, m+1},\ n>m+2,\ \quad  
{\mathbb P}_{n-1; n} = \mbox{id},\ \quad {\mathbb P}_{n; n} = {\mathcal P}_{0 \, n}$\\ and for $1\leq n \leq N$: \\ 
$\Pi = {\mathcal P}_{1 \, 2}\ {\mathcal P}_{2 \,  3}\ldots {\mathcal P}_{N-1 \, N},\  \quad\  
\check {\mathfrak R}_{n; m}= \check r_{n-1 \, n}\  \check r_{n-2 \, n-1}\ldots \check r_{m \, m+1}, \quad n>m+1$,
$ \check {\mathfrak R}_{n;n} =\mbox{id}$, $\ \check {\mathfrak R}_{n+1;n} =\check r_{n \, n+1}$. Note that $\log (\Pi)$ 
is the momentum operator of the system. Let us also define the ordered product:
\begin{equation}
\prod_{1\leq j\leq k}^{\leftarrow} \check r_{n_j  \, n_{j}+1 }  =\check  r_{n_1 \, n_1+1}  
\check r_{n_2 \, n_2+1} \ldots \check r_{n_k \, n_k+1}:\  \quad n_1>n_2>\ldots n_k. \nonumber
\end{equation}

We compute all the members of the expansion of the monodromy $T^{(k)}$, 
using the notation introduced above and the definition (\ref{mono}):
\begin{eqnarray}
&&T_0^{(N)} ={\mathbb P}_{N;1}= {\mathcal  P}_{01} \Pi,\ \nonumber\\
&&T_0^{(N-1)}=   \sum_{n=1}^N {\mathbb P}_{N;n+1} r_{0n}{\mathbb P}_{n-1;1} = \Big ( \sum_{n=1}^{N-1}\check r_{n \, n+1} 
+ \check r_{N \, 0} \Big ) {\mathcal  P}_{0 \, 1}\Pi,\ 
 \ \ldots\label{mo1}\nonumber \\
&&T_0^{(N-k)} = \sum_{1\leq n_k < \ldots < n_1 \leq N} \prod_{j=1}^k {\mathbb P}_{n_{j-1}-1; n_j+1}  r_{0 \, n_j}{\mathbb P}_{n_{k} -1;1} = 
\nonumber\\
   && \Big ( \sum_{1 \leq n_{k} <\ldots < n_1 < N} \prod_{1\leq j\leq k}^{\leftarrow} \check r_{n_j \,  n_{j} +1}  +  
\sum_{1 \leq n_{k} <\ldots < n_2 < N} \prod_{2 \leq j \leq k}^{\leftarrow} \check r_{N0} \check r_{n_j \,  n_{j}+1} \Big ) {\mathcal P}_{01} \Pi,\  
\ \ldots \label{mo2}\nonumber\\
&&T_0^{(1)}= \Big ( \sum_{n=1}^{N-1} \check r_{N0} \check {\mathfrak R}_{N; n+1} \check {\mathfrak R}_{n; 1}  + \check  
{\mathfrak R}_{N; 1} \Big ) {\mathcal P}_{0 \, 1} \Pi\nonumber\\ 
&&T_0^{(0)} =  {\mathbb R}_{N;1}= \check r_{N \, 0} \check{\mathfrak  R}_{N; 1} {\mathcal P}_{0 \, 1}  \Pi={\mathbb P}_{N;1}. \label{mof} \nonumber
\end{eqnarray}

Recall that we can express the monodromy matrix and consequently all $T^{(k)}$ in a block form, i.e. 
$T^{(k)} = \sum_{x, y \in X} e_{x,y} \otimes T_{x, y}^{(k)}$, thus 
${\mathfrak t}^{(k)} = tr_0(T_0^{(k)}) = \sum_{x \in X} T^{(k)}_{x, x}$.
%It follows then from the expressions for $T^{(k)}$ above,
%and after taking the trace over the auxiliary space (denoted by the subscript $0$), 
%that the mutually commuting quantities ${\mathfrak t}^{(k)}$ read as 
%5\begin{eqnarray}
%&&  {\mathfrak t}^{(N)} = \Pi_{1;N-1}, \\
%&& {\mathfrak t}^{(N-1)} =   \Big (\sum_{n=1}^N\check r_{nn+1}\Big )\Pi_{1;N-1},\ \ldots\\
%&& {\mathfrak t}^{(N- k)} = \  \ldots \\
%&& {\mathfrak t}^{(1)} = 
%\end{eqnarray}
${\mathfrak t}^{(k)}$ commute among each other, and hence any combination of them also provides a family of mutually commuting quantities.
For instance, we consider the following convenient combination: $H^{(k)} = {\mathfrak t}^{(k)} ({\mathfrak t}^{(N)})^{-1},\ k =1,\ldots, N-1$ and 
$H^{(N)} = {\mathfrak t}^{(N)}=  \Pi$, then (periodicity is naturally imposed after taking the trace, $N+1 \equiv 1$):
\begin{eqnarray}
H^{(N-1)} &=& \sum_{n=1}^N\check r_{n \, n+1},
\nonumber\\
H^{(N-2)} &= &\sum_{1 \leq m < n \leq N} \check r_{n \, n+1} \check r_{m \, m+1} + \sum_{n=1}^{N-2}\check r_{n \, n+1} \check r_{N \, 1}+ 
\check r_{N \, 1} \check r_{N-1 \, N},\  
\ \ldots  \label{Ham1}\nonumber \\
H^{(N-k)} &=&  \sum_{1 \leq n_{k} <\ldots < n_1 < N} \prod_{1\leq j\leq k}^{\leftarrow} \check r_{ n_j \, n_{j} +1} + \sum_{1 \leq n_{k} <\ldots < n_2 < N-1}  
\prod_{2\leq j\leq k}^{\leftarrow} \check r_{ n_j \, n_{j}+1 }\check r_{N \, 1} \nonumber\\ 
&+&  \sum_{1 \leq n_{k} <\ldots < n_2 =N-1}  \prod_{l+1\leq j\leq k}^{\leftarrow}  \check r_{ n_j \, n_{j}+1 }\ 
\check r_{N \, 1}\prod_{2\leq j\leq l}^{\leftarrow} \check r_{ n_j \, n_{j}+1 }\Big |_{c_j= 0,\ c_l >0},\, \  \ldots \nonumber\\
 H^{(1)} &=&  \sum_{n=1}^{N-1}  \check {\mathfrak R}_{n; 1} \check r_{N \, 1} \check {\mathfrak R}_{N; n+1} + 
\check  {\mathfrak R}_{N; 1}, \nonumber
\end{eqnarray}
where we define $c_j = n_j -n_{j+1 }-1$.
Indeed, the Hamiltonians  $H^{(k)}$ for $k =1, \ldots, N-1$ are expressed solely in terms of  the elements $\check r_{n \,  n+1}$ and $\check r_{N \, 1}$
and this essentially concludes our proof. Let  us also for the sake of completeness report $H^{(0)}$, which takes the simple form 
$H^{(0)} =  tr_0 \Big (\check r_{N \, 0} \check{\mathfrak  R}_{N; 1} {\mathcal P}_{0 \, 1}\Big )$, but as opposed to the rest of 
the commuting Hamiltonians $H^{(k)},\ k=1, \ldots, N-1$,  it can not be expressed only in terms of $\check r_{n \, n+1}$ and $\check r_{N \, 1}$.

It will be also  instructive for our purposes here, related especially with Proposition 4.11 (presented later in the text), to compute explicitly  
$T^{(0)}$ and ${\mathfrak t}^{(0)}$.
Recalling expression (\ref{mof}) and the form of the brace solution (\ref{brace1}) we have:
\begin{equation}
T^{(0)} = \sum_{x_1,.. ,x_N ,y_{1}, .. ,y_{N}\in X} e_{y_{N},\sigma_{x_1}(y_1) }\otimes e_{x_1, 
\tau_{y_1}(x_1)} \otimes \ldots \otimes e_{x_N, \tau_{y_N}(x_N)}
\end{equation}
and by taking the trace (periodic boundary conditions: $N+1 \equiv 1$)
\begin{equation}
{\mathfrak t}^{(0)} = \sum_{x_1,..,x_N ,y_{1},.. ,y_{N}\in X }  e_{x_1, \tau_{y_1}(x_1)}\otimes  
\ldots \otimes e_{x_N, \tau_{y_N}(x_N)}, \label{t0}
\end{equation}
where both expressions above are subject to the constraints: $y_{n} = \sigma_{x_{n+1}}(y_{n+1})$.
\end{proof}

We show below an interesting property regarding the element $ \check {\mathfrak R}_{N-1;1}$ introduced in the proof of the latter Proposition 
(see also relevant findings in connection to Murphy elements in Hecke algebras in \cite{Doikou2}).

\begin{lemma} The action of  $ \check {\mathfrak R}_{N;1} = \check r_{N-1 \,  N}\  \check r_{N-2 \,  N-1}\ldots \check r_{1 \, 2}$ 
on the elements of the $A$-type Hecke algebra is given by
\begin{eqnarray}
&&  \check {\mathfrak R}_{N;1}\ \check r_{n \, n+1}  =\check r_{n-1 \, n}\  \check {\mathfrak R}_{N;1},\ 
\quad \forall n \in \{2, \ldots, N-2 \}\nonumber\\
&&  \check {\mathfrak R}_{N;1}\  \check r_{1 \, 2} =   \check {\mathfrak R}_{N;2}, \quad    
\check r_{N-1 \,  N}\  \check {\mathfrak R}_{N;1}=   \check {\mathfrak R}_{N-1;1} \nonumber
\end{eqnarray} 
\end{lemma}
\begin{proof} The poof is straightforward via the use of the braid relation\\
$\check r_{n \, n+1} \check r_{n+1 \, n+2} \check r_{n \, n+1} =   
\check r_{n+1 \, n+2} \check r_{n \, n+1}  \check r_{n+1 \, n+2}$,  $\check r^2 = I \otimes I$,  and the form of $\check {\mathfrak R}_{N;1}$.
\end{proof}

We showed that the Hamiltonians $H^{(k)},\ k=1,\ldots N-1$ introduced in Proposition \ref{Traces} are  expressed exclusively in 
terms of the $A$-type Hecke elements and the periodic element $\check r_{N1}$. This fact will be exploited later when 
investigating the symmetries of 
the conserved quantities for various brace solutions. In the special case 
where $\check r = {\mathcal P}$, i.e. the Yangian (see also Corollary \ref{Corollary1} below),  the Hamiltonian is $\mathfrak{gl}_{\mathcal N}$  
symmetric (see also (\ref{coproduct1})).  However, if we focus on the more general brace solution we conclude that there is no 
non-commutative algebra as symmetry of the Hamiltonian or in general of the transfer matrix, with the exception of certain special
cases that will be examined later in the text. 
Note that  in \cite{DoiSmo} the existence of a non-commutative algebra that is also a symmetry of the open boundary Hamiltonian is shown 
(see also \cite{DoikouNepomechie} and references therein for relevant findings). 

The notion of  the so-called Murphy elements associated to Hecke algebras, 
emerging from open boundary transfer matrices \cite{Doikou2}, is also discussed in \cite{DoiSmo} 
for $R$-matrices that come from braces. Moreover, for a special class of set theoretical solutions and for a special choice of boundary 
conditions it is shown \cite{DoiSmo}, that not only the corresponding Hamiltonian is $\mathfrak{gl}_{\mathcal N}$ symmetric, 
but also the boundary transfer matrix.\\

Let us recall in the next corollary the known result about the existence of a non-commutative  algebra that is a symmetry 
of the transfer matrix in the Yangian case.

\begin{cor}\label{Corollary1}  In the case of Yangian ($\check r ={\mathcal P}$), $\mathfrak {gl}_{\mathcal N}$ 
is a symmetry of the periodic transfer matrix, i.e.
\begin{equation}
\Big [\Delta^{(N)} \big ( e_{i, j} \big ), \  {\mathfrak t}(\lambda)\Big ] =0,\quad i,\ j \in \{1,\ 2, \ldots, {\mathcal N}\}.
\end{equation}
\end{cor}
\begin{proof}
The monodromy matrix satisfies the RTT relation (\ref{RTT}), also in this case $r=I \otimes I$, 
then $T^{(0)} = I^{\otimes (N+1)}$ (recall the expansion $T(\lambda) = \sum_{n}\lambda^{-n} T^{(n)}$),
and  from (\ref{fund2b}) for $n=1$ it follows that
\begin{equation}
\Big [T_1^{(1)}, \  T^{(m)}_2\Big ] = T_{2}^{(m)}{\mathcal P}_{12} - {\mathcal P}_{12} T_{2}^{(m)}.
\end{equation}
By taking the trace over the second space we conclude
\begin{equation}
\Big [T_1^{(1)}, \  \mathfrak{t}^{(m)} \Big ] =0\ \Rightarrow\  \Big [T_{i,j}^{(1)}, \  {\mathfrak t}(\lambda)\Big ] =0.
\end{equation}
$T_{i,j}^{(1)} \in \mbox{End}\big ( ({\mathbb C}^{\mathcal N})^{\otimes N}\big)$ are the entries of the $T$ matrix,
and  are tensor representations of the $\mathfrak {gl}_{\mathcal N}$ algebra  (\ref{gln}), 
\begin{equation}
T_{i,j}^{(1)} = \Delta^{(N)} \big ( e_{i, j} \big ), \quad i,\ j \in \{1,\ 2, \ldots,  {\mathcal N}\},
\end{equation}
where the co-product is defined in (\ref{coproduct1}) (${\mathfrak l}_{i,j} \mapsto e_{i,j},$  recall also  Remark \ref{remg}), 
i.e. the transfer matrix enjoys the $\mathfrak {gl}_{\mathcal N}$ symmetry. 
\end{proof}
 Let $(X, \check r)$ be  a set theoretical solution of the Yang-Baxter equation. 
The next natural step is to investigate the existence 
of non-commutative algebras that are symmetries of the general open boundary transfer matrix, i.e. 
generalize Corollary \ref{Corollary1} for any brace solution. This is a fundamental problem and is investigated in \cite{DoiSmo}.

We are now in the position to provide new examples of symmetries of periodic transfer matrices constructed from 
set-theoretic solutions of the Yang-Baxter equation. 
We will use the following known fact:

\begin{lemma}\label{useful}
Let $R$ be a solution of the Yang-Baxter equation, and $B$ be a ${\mathcal N}\times {\mathcal N}$  
matrix such that
\begin{equation} 
(B\otimes B)R(\lambda)=R(\lambda)(B\otimes B). \label{Rsym1}
\end{equation}
Then
\begin{equation}
(B\otimes B^{\otimes N})T(\lambda)=T(\lambda)(B\otimes B^{\otimes N}), \label{Tsym1}
\end{equation}
where $T(\lambda)$ is the monodromy matrix.
 \end{lemma}
\begin{proof}
To prove this it is convenient to employ the index notation. Indeed, in the index notation expression (\ref{Rsym1})
is translated into: $\forall n \in \{1,\ \ldots,\ N \}$
\begin{eqnarray}
&& B_0 B_n R_{0 \, n}(\lambda) = R_{0 \, n}(\lambda) B_0 B_n  \ \Rightarrow  \nonumber\\ 
&& B_0 B_1 \ldots B_N  R_{0 \, N}(\lambda) \ldots R_{0 \, 1}(\lambda) = 
R_{0 \, N}(\lambda) \ldots R_{0 \, 1}(\lambda) B_0 B_1 \ldots B_N. \nonumber
\end{eqnarray}
The latter expression is equivalent to (\ref{Tsym1}).
\end{proof}

We will also use the following obvious  fact: 

\begin{lemma}\label{BB}If $B$ is an ${\mathcal N}\times {\mathcal N}$ matrix and 
${\mathcal P}=\sum_{1\leq i,j \leq {\mathcal N}}e_{i,j}\otimes e_{j,i}$ then \[(B\otimes B){\mathcal P}={\mathcal P}(B\otimes B).\]
\end{lemma}
 
\begin{pro}\label{cor7} Let $(X, \check r)$ be a set-theoretic solution of the braid equation and  let $f: X \to X$ be 
an isomorphism of solutions, so $f(\sigma_x(y)) = \sigma_{f(x)}(f(y))$ and $f(\tau_{y}(x) )=\tau_{f(y)}(f(x))$.
If $M = \sum_{x\in X} \alpha_x  e_{x, f(x)},$ such that $0\neq a_x \in {\mathbb C}$
and $\alpha_x \alpha_y = \alpha_{\sigma_x(y)} \alpha_{\tau_y(x)},$ $\forall x, y\in X,$ then   
\begin{equation}
 \Big [ M^{\otimes N},\  {\mathfrak t}(\lambda) \Big ]= 0,
\end{equation}
where  $\mathfrak{t}(\lambda )$ is the transfer matrix for $R(\lambda)= {\mathcal P}+\lambda {\mathcal P}\check r$.
\end{pro}
\begin{proof} By means of  
Lemmata \ref{useful} and \ref{BB} it suffices to show that
\begin{eqnarray}
&& (M \otimes M )  r=   r (M\otimes M), \label{a}
\end{eqnarray}
where recall $r =  {\mathcal P} \check r$. Indeed, then by direct computation the LHS of (\ref{a}) reduces to:
\begin{equation}
 (M \otimes M )  r= \sum_{x, y \in X} \alpha_x \alpha_y\  e_{y, f(\sigma_x(y))} \otimes e_{x, f(\tau_y(x))} \label{l1}
\end{equation}
whereas the RHS gives:
\begin{equation}
r (M\otimes M) =  \sum_{x, y \in X}  \alpha_{\sigma_x(y)} \alpha_{\tau_y(x)}\  e_{y, f(\sigma_x(y))} \otimes e_{x, f(\tau_y(x))}. \label{r1}
\end{equation}
Comparing (\ref{l1}), (\ref{r1}) we arrive at  (\ref{a}),  provided that $\alpha_x \alpha_y = \alpha_{\sigma_x(y)} \alpha_{\tau_y(x)}$.

Having shown (\ref{a}) it then immediately follows from Lemmata \ref{useful} and \ref{BB} that 
\begin{eqnarray}
 (M\otimes M^{\otimes N}) T(\lambda) = T(\lambda) (M \otimes M^{\otimes N} ). \label{symmetryT}
\end{eqnarray}
From the latter equation we focus on each element of the matrices 
(LHS vs RHS in (\ref{symmetryT})) on the auxiliary space (recall the notation $T =\sum_{x, y} e_{x,y} \otimes  T_{x,y}$):
\begin{eqnarray}
&& \alpha_{x} M^{\otimes N} T_{f(x), f(y)} =  \alpha_y T_{x,y} M^{\otimes N} \nonumber\ \Rightarrow \\ 
&&  M^{\otimes N} T_{f(x), f(x)} =  T_{x,x} M^{\otimes N}\ \Rightarrow\ \Big [ M^{\otimes N},\  {\mathfrak t}(\lambda) \Big ]= 0.
\end{eqnarray}
We can explicitly express $M^{\otimes N}$ as
\begin{equation}
M^{\otimes N} = \sum_{x_1, x_2, ..,x_N \in X }\prod_{k=1}^N{ \alpha_{x_k}}e_{x_1,f (x_1)}\otimes e_{x_2,f (x_2)}\otimes
\ldots \otimes e_{x_N,f (x_N)}.
\end{equation}
\end{proof}

\noindent Notice also, that due to the symmetry of the $R$-matrix (\ref{a}), one easily shows that if $T$ 
satisfies the RTT relation then so $MT$ does. If $M$ is non-singular then $M^{\otimes N}$ is a similarity transformation that 
leaves the transfer matrix invariant, and naturally provides information on the multiplicities of the spectrum.

A special case of the proposition above is the obvious choice: $f(x) =x,\ \forall x \in X$.
Also,  we obtain the following as immediate corollaries.

\begin{cor}
 Let $(X, {\check r})$ be a finite, non-degenerate involutive  set-theoretic solution of the Yang-Baxter equation, and let $Q_{1}, \ldots, Q_{k}$ 
be all the orbits of $X$. Let $\alpha _{1}, \ldots , \alpha _{k}\in \mathbb C$. If $M=\sum _{j=1}^{k}\alpha_{j} M_{j}$, where 
$M_{j}=\sum_{i\in Q_{j}}e_{i,i}$, then
\begin{equation}
\Big [ M^{\otimes N},\  {\mathfrak t}(\lambda) \Big ]= 0,
\end{equation}
where  $\mathfrak{t}(\lambda)$ is the transfer matrix for $R(\lambda)={\mathcal P} +\lambda {\mathcal P} \check r$.
\end{cor}
\begin{proof} If all $\alpha _{x}\neq 0$ then the result follows from Proposition \ref{cor7} (with $f(x)=x$).
 Observe that \[M^{\otimes N}=\sum_{i_{1}, \ldots , i_{N}\in \mathbb N}M_{i_{1}, \ldots , i_{N}}\alpha _{1}^{i_{1}}\cdots \alpha _{N}^{i_{N}}.\]
 We know that for all non-zero choices of $\alpha _{x}$  the matrix $M^{\otimes N}$ commutes with  ${\mathfrak t}(\lambda)$. 
By a ``Vandermonde matrix argument'' 
 each matrix $M_{i_{1}, \ldots , i_{N}\in \mathbb N}$ commutes with ${\mathfrak t}(\lambda)$, which concludes the proof.
\end{proof}

\begin{cor}
Let $(X, {\check r})$ be a finite, non-degenerate involutive  set-theoretic solution of the Yang-Baxter equation, 
and let $G(X,r)$ be its  structure group 
(i.e. the group generated by elements from $X$ and their inverses subject to relations $xy=uv$ whenever 
$r(x,y)=(u,v)$). Let $\alpha :G(X,r)\rightarrow {\mathbb C}^{*}$ 
be a group homomorphism. If $M=\sum_{x\in X}\alpha_x e_{x,x}$, then  $M^{\otimes N}$ 
commutes with the transfer matrix.
\end{cor}

Similarly, as in Proposition \ref{cor7} we use Lemmata \ref{useful} and \ref{BB} in the following  proofs. 

\begin{pro}\label{fundamental} Let $(X,{\check r})$ be a finite,  non degenerate 
	involutive set-theoretic solution of the braid equation. Let ${\mathcal N}$ be the cardinality of $X$. 
Let $x_{1}, \ldots , x_{\alpha}\in X$ for some $\alpha \in \{1, \ldots, {\mathcal N} \}$. 
Assume that ${\check r}(x_{i},y)=(y,x_{i}),$ $\forall\ y\in X$ and for all $\forall\ i \in \{1,\ldots, \alpha \}$.  Then $\forall\ i,j \in \{1,\ldots, \alpha \}$:
\begin{eqnarray}
 && \Big [ \Delta^{(N)}(e_{x_i, x_j}),\  {\mathfrak t}(\lambda) \Big ]= 0,  \label{th1}
%&& \Big [ e_{x_i, x_j}^{\otimes N},\ \mathfrak{t}(\lambda) \Big  ]=0.  \label{th2}
\end{eqnarray}
where  $\mathfrak{t}(\lambda)$ is the transfer matrix for $R(\lambda)={\mathcal P} +\lambda {\mathcal P} \check r$.
\end{pro}
\begin{proof} We first recall that $\Delta^{(N)}(e_{x_i, x_j})$ is defined in (\ref{coproduct1}).
 We will  show that $\forall  i,j \in \{1,\ldots, \alpha \}$,
\begin{eqnarray}
 \Delta(e_{x_i, x_j}) r=   r  \Delta(e_{x_i, x_j}) .  \label{a1} 
\end{eqnarray}
Indeed, by direct computation:
$\Delta(e_{x_i, x_j}) r=  \Delta(e_{x_i, x_j})=  r  \Delta(e_{x_i, x_j})$, i.e. (\ref{a1}). 
It then immediately follows from (\ref{a1})  that 
\begin{eqnarray}
 \Delta^{(N+1)}(e_{x_i, x_j}) T(\lambda) = T(\lambda)  \Delta^{(N+1)}(e_{x_i, x_j}). \label{symmetryT2}
\end{eqnarray}
Recalling on the definition of the co-product (\ref{coproduct1}),  and the notation 
$T(\lambda) = \sum_{x,y \in X} e_{x,y} \otimes T_{x,y}(\lambda)$
we conclude that expression (\ref{symmetryT2})  leads to
\begin{eqnarray}
&&\sum_w e_{x_i,w} \otimes T_{x_j,w} + \sum_{z,w} e_{z,w}\otimes \Delta^{(N)}(e_{x_i,x_j})T_{z,w}(\lambda)= \nonumber\\
&& \sum_w e_{z,x_j} \otimes T_{z,x_i} + \sum_{z,w} e_{z,w}\otimes T_{z,w}(\lambda) \Delta^{(N)}(e_{x_i,x_j}). \label{cosym2}
\end{eqnarray}
We focus on the diagonal entries of the latter expression and we obtain:
\begin{eqnarray}
&& \Big[ \Delta^{(N)}(e_{x_i,x_j}),\ T_{x_i,x_i}(\lambda)\Big ]  = -T_{x_j,x_i}(\lambda) + \delta_{ij}T_{x_j,x_i}(\lambda),  \nonumber\\
&&   \Big[ \Delta^{(N)}(e_{x_i,x_j}), \ T_{x_j,x_j}(\lambda)\Big ]  = T_{x_j,x_i}(\lambda) -  \delta_{ij}T_{x_j,x_i}(\lambda), \nonumber\\
&& \Big[ \Delta^{(N)}(e_{x_i, x_j}),\  T_{z,z}(\lambda)\Big ]   =0, \quad z \neq x_i,\ x_j. \label{cosym3}
\end{eqnarray}
Summing up all the terms above we arrive at  (\ref{th1}),
$\Big [ \sum_{x\in X} T_{x, x}(\lambda),\  \Delta^{(N)}(e_{x_i,x_j}) \Big ] =0, \quad \forall i, j \in \{1, \ldots, \alpha\},$
i.e. we conclude that the transfer matrix is $\mathfrak{gl}_{\alpha}$ symmetric.
\end{proof}

\begin{lemma}\label{BC} If $r = {\mathcal P}\check r$ and
\begin{eqnarray}
 \Delta(e_{x_i, x_j}) r=   r  \Delta(e_{x_i, x_j})  \label{a3} 
\end{eqnarray}
for some $x_i,\  x_j \in X$, then 
\begin{equation}
e_{x_i,x_j} \otimes e_{x_i,x_j}  r = r e_{x_i,x_j} \otimes e_{x_i,x_j}  \label{k1}
\end{equation}
and also, $\forall n \in \{1, \ldots, N\}$:
\begin{equation}
\Big [\sum_{m_1<m_2... <m_n=1}^N (e_{x_i, x_j})_{m_1} \ldots (e_{x_i, x_j})_{m_n},\  \mathfrak{t}(\lambda) \Big ]= 0,\label{k2}
\end{equation}
where we use the standard index notation,
\begin{equation}
(e_{x_i, x_j})_{m} = I \otimes \ldots I \otimes \underbrace{e_{x_i, x_j}}_{m^{th}\ \mbox{ position}} \otimes I \ldots \otimes I,
\end{equation}
 where $I$ appears $N-1$ times.
\end{lemma}
\begin{proof}
The proof is straightforward: from the definition of the co-product $\Delta$,  
the fact that $e_{x,y}^2 = e_{x,y} \delta_{xy}$ and also $ \big ( \Delta(e_{x_i, x_j}) \big )^2 r=  r  \big (\Delta(e_{x_i, x_j}) \big )^2$, 
we arrive at (\ref{k1}). 

Similarly as in the proof of  Proposition \ref{fundamental} we know that if $\Delta(e_{x_i, x_j}) r=   r  \Delta(e_{x_i, x_j}) $,
then 
$\Big [  \Delta^{(N)}(e_{x_i, x_j}),\  \mathfrak{t}(\lambda)  \Big ]= 0$, and  
$\Big [  \big(\Delta^{(N)}(e_{x_i, x_j})\big )^n,\  \mathfrak{t}(\lambda)  \Big ]= 0$, $n \in \{1, \ldots, N\}$,
which together with  $e_{x,y}^2 = e_{x,y} \delta_{xy}$ lead to (\ref{k2}).
\end{proof}

We  obtain in what follows a  general class of symmetries, associated to solutions endowed with some extra special properties.
Recall that solutions $(X,r)$  such that $r(x,x)=(x,x),$ $\forall\ x \in X$ are called square free, and  they were introduced by 
Gateva-Ivanova. There was a famous conjecture by Gateva-Ivanova as to whether or not these solutions need to
have a finite multi-permutation level.  
Vendramin subsequently showed that this was not necessary \cite{Leandro}. Later Ced{\' o}, 
Jespers and Okni{\' n}ski investigated these solutions using wreath  products of groups and also gave many interesting 
examples of square-free solutions.

\begin{pro}\label{88} Let $(X, {\check r})$ be a finite  non degenerate involutive set-theoretic solution of 
the braid equation. Let ${\mathcal N}$ be the cardinality of $X$. Let $x_{1}, \ldots , x_{\alpha}\in X$ for some $\alpha \in \{1, \ldots, {\mathcal N}\}$ 
be such that ${\check r}(x_j, x_j)=(x_j, x_j)$ $\forall i, j \in \{1, \ldots, \alpha\}$.
Let ${\mathfrak t}(\lambda) =\lambda^N \sum_{k=0}^{N} {\mathfrak t}^{(k)}\lambda ^{-k}$ be the 
transfer matrix for $R(\lambda)={\mathcal P}+\lambda {\mathcal P}{\check r}$.
Then  $\forall i, j \in \{1, \ldots, \alpha\}$ and $k=1, \ldots , N,$ 
\[\Big [e_{x_i, x_j}^{\otimes N},\ {\mathfrak t}^{(k)}\Big ] =0. \]
\end{pro}

\begin{proof}
We first show by direct computation that:
\begin{eqnarray}
&& e_{x_i,x_j} \otimes e_{x_i,x_j} \check r = \check r e_{x_i,x_j} \otimes e_{x_i,x_j}\ \Rightarrow\  \nonumber\\
&& e_{x_i, x_j}^{\otimes N} \check r_{n \, n+1}= \check r_{n \, n+1}e_{x_i, x_j}^{\otimes N}, \quad n = 1, \ldots, N-1, \label{cc1} \\
&& e_{x_i, x_j}^{\otimes N} \check r_{1 \, N}= \check r_{1 \, N}e_{x_i, x_j}^{\otimes N} \label{cc2}
\end{eqnarray}
By multiplying equation (\ref{cc2}) with ${\mathcal P}_{1 \, N}$ we conclude that: $\Big[e_{x_i,x_j}^{\otimes N},\ \check r_{N \, 1} \Big ]= 0$.

Then recall from Proposition \ref{Traces} that $H^{(k)} = {\mathfrak t}^{(k)} ({\mathfrak t}^{(N)})^{-1}$ for $k=1, 2 \ldots, N-1$ 
are expressed exclusively in terms of $\check r_{n \, n+1},~~n=1, \ldots,  N-1$ and $\check r_{N \, 1}$.  
Recall also from the proof of Proposition \ref{Traces} that
$H^{(N)} = {\mathfrak t}^{(N)} = {\mathcal P}_{1 \, 2} {\mathcal P}_{2 \, 3} \ldots {\mathcal P}_{N-1 \, N}$, 
which immediately leads to (from the definition of ${\mathcal P}$):
$\Big[ e_{x_i,x_j}^{\otimes N},\  H^{(N)} \Big ]=0$, which in turn together with 
expressions (\ref{cc1}), (\ref{cc2}) and Proposition \ref{Traces} lead to:
$\Big [ e_{x_i, x_j}^{\otimes N},\ H^{(k)} \Big ] =0, \quad k =1, \ldots, N$.
And due to $H^{(k)} = {\mathfrak t}^{(k)} ({\mathfrak t}^{(N)})^{-1}$ for $\ k =1, \ldots N-1$ and  $ H^{(N)} = {\mathfrak t}^{(N)}$, 
we conclude  that $\Big [ e_{x_i, x_j}^{\otimes N},\  {\mathfrak t}^{(k)}\Big ] =0$ for $\ k =1, \ldots, N$ and $i,\ j =1, \ldots, \alpha$. 
\end{proof}

Recall that by Remark \ref{nilpotent} every nilpotent ring is a brace when we define $a\circ b=ab+a+b$ 
(we call this brace the corresponding brace).

We restrict  our attention in what follows in the case where $N$ is odd.
Let $(B, +, \circ )$ be a brace. We say that $a\in B$ is central if  $a\circ b=b\circ a,$  $\forall\ b\in B$. 
 
\begin{pro} Let $(B, +, \cdot )$ be a nilpotent ring and $(B, +, \circ )$ be the corresponding brace. 
 Let $a$ be a central element in $B$, and $a+a=0$, $a\circ a=0$. Let $X\subseteq B$ 
 and let ${\check r}_{B}$ be defined as in Theorem \ref{Rump}. Let  $x,y\in X$ and $x=\sigma _{b}(a)=ba+a$, $y=\sigma _{c}(a)=ca+a $ for some $b,c\in B$. 
Then 
 \begin{equation}
 \Big [ e_{x,y}^{\otimes N},\  {\mathfrak t}(\lambda) \Big ]= 0,
\end{equation}
\end{pro}
\begin{proof} 

First recall that ${\mathfrak t}(\lambda) = \lambda^N \sum_{k=0}^N {{\mathfrak t}^{(k)} \over \lambda^k }$, 
recall also that we can express the monodromy matrix and consequently all $T^{(k)}$ in a block form, i.e. 
$T^{(k)} = \sum_{z, w \in X} e_{z,w} \otimes T_{z, w}^{(k)}$, thus ${\mathfrak t}^{(k)} = tr_0(T_0^{(k)}) = \sum_{z \in X} T^{(k)}_{z, z}$.
Via Proposition \ref{88} it suffices to show that $e_{x, y}^{\otimes N}$ commutes with  ${\mathfrak t}^{(0)}$.  

Denote \[W_{p_{1}, \ldots , p_{N}}=Q_{p_{1},p_{2}}\otimes Q_{p_{2},p_{3}}\otimes \cdots \otimes Q_{p_{N-1},p_{N}}\otimes Q_{p_{N},p_{1}}, \]
where $Q_{i,p}=\sum_{j\in W_{i,p}}e_{j,j^{i}}$ and $W_{i,p}=\{j: {{ }^ji}=p\}.$ 
 Recall also the form of ${\mathfrak t}^{(0)}$ from the proof of Proposition \ref{Traces}, expressed as
 \[{\mathfrak t}^{(0)}=\sum _{p_{1}, \ldots , p_{N}\leq {\mathcal N}} W_{p_{1}, \ldots , p_{N}}=
Q_{p_{1},p_{2}}\otimes Q_{p_{2}, p_{3}}\otimes \cdots \otimes Q_{p_{N-1}, p_{N}}\otimes Q_{p_{N}, p_{1}},\]
 where ${\mathcal N}$ is the cardinality of $X$.
 We will show that \[e_{x,y}^{\otimes N}W_{p_{1}, \ldots , p_{N}}=W_{p_{1}, \ldots , p_{N}}e_{x,y}^{\otimes N},\]
 for all $p_{1}, \ldots , p_{N}\leq {\mathcal N}$.

{\bf Part 1}. 
We will first calculate $e_{x,y}^{\otimes N}W_{p_{1}, \ldots , p_{N}}$.

Notice that $e_{x,y}^{\otimes N}W_{p_{1}, \ldots , p_{N}}=e_{x,y}Q_{p_{1},p_{2}}\otimes e_{x,y}Q_{p_{2},p_{3}}
\otimes \cdots \otimes e_{x,y} Q_{p_{N}, p_{1}}.$
If it is non zero then $e_{x,y}Q_{p_{i}, p_{i+1}}\neq 0$ for every $i$.

Notice that  $e_{x,y}Q_{p_{i}, p_{i+1}}=\sum_{j\in W_{p_{i}, p_{i+1}}}e_{x,y}e_{j,j^{p_{i}}}$. If $e_{x,y}e_{j,j^{p_i}}\neq 0$ then $j=y$, and  
 $j^{p_{i}}=y^{p_{i}}$, hence \[e_{x,y}Q_{p_{i}, p_{i+1}}=e_{x,y^{p_{i}}}.\] Notice that  $y=j\in   W_{p_{i}, p_{i+1}}$, 
hence ${{ }^y{p_{i}}}=p_{i+1}$. Similarly
 ${{ }^y{p_{i+1}}}=p_{i+2}$, this implies ${{ }^{y\circ  y}{p_{i}}}=p_{i+2}$. 
 Observe that $y\circ y=0$, and so $p_{i}=p_{i+2}$ for every $i$ (where $p_{i+N}=p_{i}$).
  Since $N$ is odd  it follows that  
 $p_{1}=p_{2}=\ldots =p_{N}$ and ${ }^yp_{i}=p_{i},$  $\forall i$.

 Now ${{ }^y{p_{i}}}=p_{i+1}$ implies ${{ }^yp_{i}}=p_{i}$,  hence ${{ }^ap_{i}}=p_{i}$. It follows that $y^{p_{i}}=y$ (since $a$ is central in $B$).
 Therefore, $e_{x,y}^{\otimes N}W_{p_{1}, \ldots , p_{N}}=e_{x,y}^{\otimes N}$ provided that $p_{1}=\ldots =p_{N}$ and 
 ${{ }^ap_{1}}=p_{1}$, and otherwise it is zero.

$ $
{\bf Part 2}. 
We will now calculate $W_{p_{1}, \ldots , p_{N}}e_{x,y}^{\otimes N}$. 

Notice that  if $W_{p_{1}, \ldots , p_{N}}e_{x,y}^{\otimes N}\neq 0$ then $Q_{p_{i}, p_{i+1}}e_{x,y}\neq 0$, for every $i$.  
Observe that $Q_{p_{i}, p_{i+1}}e_{x,y}=\sum_{j\in W_{p_{i}, p_{i+1}}}e_{j,j^{p_{i}}}e_{x,y}$. If $e_{j,j^{p_{i}}}e_{x,y}\neq 0$ 
then  $j^{p_{i}}=x$. Hence $j=x^{q_{i}}$ where $q_{i}$ is the inverse of $p_{i}$ in the group $(B, \circ )$. It follows that $j=q_{i}a+a$. 
Therefore, for given $p_{i}$  element $j_{i}$ is uniquely determined.
Consequently, $Q_{p_{i}, p_{i+1}}e_{x,y}=e_{j_{i}, y}$ and $j_{i}= {j\in W_{p_{i}, p_{i+1}}}$. 
This implies ${{ }^{j_{i}}{p_{i}}}= p_{i+1}$, similarly ${{ }^{j_{i+1}}{p_{i+1}}}= p_{i+2}$, etc.  It follows that  
${ }^j{p_{i}}= p_{i}$ where $j=j_{i+N-1}\circ \cdots \circ  j_{i+1}\circ j_{i}$. Observe that $j=a+qa$ for some $q\in B$ 
 by assumptions on $x$.  Previously it was shown that ${{ }^{j_{i}}{p_{i}}}= p_{i+1}=p_{i}$, hence ${{ }^{x}{p_{i}}}=p_{i}$. 
Therefore ${{ }^a{p_{i}}}=p_{i}$ for $i=1, \ldots N$. This implies $x^{p_{i}}=x$. Observe that $j_{i}=x^{q_{i}}$ notice 
that $q_{i}=p_{i}\circ \cdots \circ p_{i}$ and since $x^{p_{i}}=x$ it follows that $j_{i}=x$.
 
 Therefore, $W_{p_{1}, \ldots , p_{N}}e_{x,y}^{\otimes N}=e_{x,y}^{\otimes N}$ provided that $p_{1}=\ldots =p_{N}$ and 
 ${{ }^ap_{1}}=p_{1}$, and otherwise it is zero.
 Consequently, $e_{x,y}^{\otimes N}W_{p_{1}, \ldots , p_{N}}=W_{p_{1}, \ldots , p_{N}}e_{x,y}^{\otimes N}.$
\end{proof}

%%%%%%%%%%%%%%%%%%%

\subsection*{Acknowledgments}

\noindent  We are indebted to Robert Weston for illuminating discussions.  AS thanks 
Robert Weston for explaining to her several  topics from quantum integrable systems. We also thank Vincent Caudrelier,
Ferran Ced{\' o}, Eric Jespers, Jan Okni{\' n}ski  and Michael West for useful suggestions on the manuscript. 
AD  acknowledges support from the EPSRC research grant EP/R009465/1, 
AS acknowledges support from the EPSRC programme grant EP/R034826/1.
 The authors also acknowledge support from the EPSRC research grant EP/V008129/1.


\begin{thebibliography}{99}

\bibitem{ABS}
V.E. Adler, A.I. Bobenko and Yu.B. Suris, {\em Classification of integrable equations on quad-graphs.
The consistency approach}, Comm. Math. Phys. 233 (2003) 513.

\bibitem{Ba}  
D. Bachiller, {\em Counterexample to a conjecture about braces}, J. Algebra, 453 (2016) 160--176. 
 
\bibitem{bcjo} 
D. Bachiller, F. Ced{\'o}, E. Jespers and J. Okni{\' n}ski,{\em Iterated matched products of finite braces and simplicity; 
new solutions of the Yang-Baxter equation}, Trans. Amer. Math. Soc. 370 (2018) 4881--4907.

\bibitem{Baxter}
R.J. Baxter, {\em Exactly solved models in statistical mechanics}, Academic Press (1982).

\bibitem{Baz}
V.V. Bazhanov and S.M. Sergeev, {\em Yang–Baxter maps, discrete integrable equations and quantum groups},   Nucl. Phys. B926 (2018) 509-543 .

\bibitem{BerKaz}
A. Berenstein and D. Kazhdan, {\em Geometric and unipotent crystals}, in Visions in
13 Mathematics: GAFA special volume (2000) 1p. 88.

\bibitem{Tomasz} 
T. Brzezi{n}ski, {\em Trusses: Between braces and ring}, Trans. Amer. Math. Soc. 372 (2019) 4149--4176.

\bibitem{Catino} 
F. Catino, I. Colazzo and P. Stefanelli, {\em Semi-braces and the Yang-Baxter equation}, J. Algebra, 483 (2017) 163--187. 

\bibitem{fc}
 F. Ced{\' o}, {\em Left braces: solutions of the Yang-Baxter equation}, Adv. Group Theory Appl., Vol. 5 (2018) 33--90.

\bibitem{[6]}
F. Ced{\' o}, E. Jespers, and J. Okni{n}ski, {\em  Braces and the Yang-Baxter equation}.
Comm. Math. Phys., 327(1) (2014) 101–116.

\bibitem{Doikou2}
A. Doikou, {\em Murphy elements from the double-row transfer matrix}, J. Stat. Mech. (2009) L03003.

\bibitem{DoikouNepomechie}
A. Doikou and R.I. Nepomechie, {\em Bulk and Boundary S Matrices for the SU(N) Chain}, Nucl. Phys. B521 (1998) 547-572.

\bibitem{DoiSmo}
A. Doikou and A. Smoktunowicz, {\em Set-theoretical Yang-Baxter and reflection equations 
$\&$ quantum group symmetries}, Lett. Math. Phys. 111 (2021) 105.

\bibitem{Drin} 
V.G. Drinfeld, {\em On some unsolved problems in quantum group theory}, Lecture Notes in Math., 
vol. 1510, Springer-Verlag, Berlin (1992) 1-8. 

\bibitem{Drinfeld}
V.G. Drinfel’d, {\em  Hopf algebras and the quantum Yang–Baxter equation}, Soviet. Math. Dokl.
32 (1985) 254;\\
{\it A new realization of Yangians and quantized affine algebras}, Soviet. Math. Dokl. 36 (1988) 212.

\bibitem{ESS}  
P. Etingof, T. Schedler and A. Soloviev, {\em Set-theoretical solutions to the quantum Yang–Baxter equation}, 
Duke Math. J. 100 (1999) 169--209.

\bibitem{Eti}
P. Etingof, {\em Geometric crystals and set-theoretical solutions to the quantum Yang-Baxter equation},
Comm. Algebra 31 (2003) 1961.

\bibitem{FadTakRes}
L.D. Faddeev, N.Yu. Reshetikhin and L.A. Takhtajan, {\em Quantization of Lie groups and Lie
algebras}, Leningrad Math. J. 1 (1990) 193.

\bibitem{FT} 
L.D. Faddeev and L.A. Takhtajan, {\em Hamiltonian Methods in the Theory of Solitons}, Springer-Verlag (1987).

\bibitem{15} 
T. Gateva-Ivanova, {\em  Quadratic algebras, Yang-Baxter equation, and Artin--Schelter regularity}, 
Adv. in Math. 230 (2012) 2152--2175.

\bibitem{13} 
T. Gateva-Ivanova, {\em Skew polynomial rings with binomial relations}, J. Algebra 185 (1996)
710–753.

\bibitem{GI3} 
T. Gateva-Ivanova, {\em Binomial skew polynomial rings, Artin--Schelter regularity, 
and binomial solutions of the Yang-Baxter equation},  Serdica Mathematical Journal 
Volume: 30, Issue: 2-3 (2004 )431--470.

\bibitem{Gateva} 
T. Gateva-Ivanova, {\em A combinatorial approach to noninvolutive set-theoretic solutions of the Yang--Baxter equation}, 
preprint (2018), a {\tt arXiv:1808.03938 [math.QA]} (2018).

\bibitem{gateva}  
T. Gateva-Ivanova, {\em Set-theoretic solutions of the Yang–Baxter equation, braces and symmetric groups}, 
Adv. Math., 388(7) (2018) 649--701.

\bibitem{20}
T. Gateva--Ivanova and M. Van den Bergh, {\em Semigroups of I-type},  J.  Algebra 206 (1997) 97--112. 

\bibitem{GV} 
 L. Guarnieri and L. Vendramin, {\em Skew braces and the Yang–Baxter equation},  
Math. Comp.  86(307) (2017) 2519–2534.

\bibitem{HatKunTak}
G. Hatayama, A. Kuniba and T. Takagi, {\em Soliton cellular automata associated with crystal bases},
Nucl. Phys. B577 (2000) 619.

\bibitem{Hienta}
J. Hietarinta, {\it Permutation-type solutions to the Yang-Baxter and other nsimplex
equations}, J. Phys. A30 (1997) 4757-4771.

\bibitem{APil}  
P. Jedlicka, A. Pilitowska and A. Zamojska-Dzienio, {\em The retraction relation for biracks}, J. Pure Appl. Algebra 223 (2019) 3594--3610.

\bibitem{JKA} 
E. Jespers, E. Kubat and A. Van Antwerpen,  {\em The structure monoid and algebra of a non-degenerate 
set-theoretic solution of the Yang-Baxter equation},  Trans. Amer. Math. Soc. 372 (2019) 7191--7223.

\bibitem{kava}  E. Jespers, E. Kubat, A. Van Antwerpen and L. Vendramin, {\em  Factorizations of skew braces},  Math. Ann. 375 (2019) no. 3-4, 1649--1663.

\bibitem{24} 
E. Jespers and J. Okni{\' n}ski, {\em Monoids and groups of I-type}, 
Algebr. Represent. Theory 8 (2005) 709--729. 

\bibitem{jk} E. Jespers and J. Okni{\' n}ski, {\em  Binomial semigroups}, Journal of Algebra, 202 (1998) 250--275.

\bibitem{EJJO} E. Jespers and J. Okni{\' n}ski,  {\em Noetherian semigroup algebras,} Series: Algebra and Applications, Vol. 7,  (2007).

\bibitem{25}
 E. Jespers, J. Okni{\' n}ski, and M. Van Campenhout, {\em Finitely generated algebras defined by homogeneous 
quadratic monomial relations and their underlying monoids}, J.  Algebra 440 (2015) 72--99.

\bibitem{26}
 E. Jespers and M. Van Campenhout, {\em Finitely generated algebras defined by homogeneous quadratic 
monomial relations and their underlying monoids II}, J.  Algebra 492 (2017) 524--546. 


\bibitem{Jimbo}
M. Jimbo, {\em A q-difference analogue of U(g) and the Yang-Baxter equation},  Lett. Math. Phys. 10 (1985) 63.

\bibitem{Jimbo2}
M. Jimbo, {\em Quantum R-matrix for the generalized Toda system}, Comm. Math. Phys.102 (1986) 537–547.

\bibitem{Ka}  
L.H. Kauffman, {\em Virtual knot theory}, European Journal of Combinatorics 20 (1999) 663--691.

\bibitem{KSV} 
A. Konovalov, A. Smoktunowicz and L. Vendramin, {\em On skew braces and their ideals},  Experimental Mathematics,  Volume 30 (2021) 95--104.

\bibitem{IL} 
I. Lau, {\em Left Brace With The Operation * Associative Is A Two-sided Brace}, to appear in J. Algebra and its Applications.

\bibitem{Lebed}  
V. Lebed and L. Vendramin, {\em On structure groups of set-theoretical solutions to the Yang-Baxter equation},
 Proc. Edinburgh Math. Soc. ,Volume 62,  Issue 3 (2019)  683 --717.

\bibitem{yangians} 
A. Molev, M. Nazarov and  G. Olshanski, {\em Yangians and classical Lie algebras}, Russ. Math. Surveys
51 (1996) 205. 

\bibitem{Papag}
V.G. Papageorgiou, Yu.B. Suris, A.G. Tongas and A.P. Veselov, {\em On quadrirational Yang-Baxter
Maps}, SIGMA 6 (2010) 033.

 \bibitem{[25]}
W. Rump, {\em A decomposition theorem for square-free unitary solutions of the quantum Yang-Baxter equation}. Adv. Math., 193(1) (2005) 40–55. 

\bibitem{[26]}
W. Rump, {\em  Braces, radical rings, and the quantum Yang-Baxter equation}, 
J. Algebra 307(1) (2007) 153–170. 

\bibitem{SVB} 
A. Smoktunowicz and L. Vendramin, {\em On Skew Braces (with an appendix by N. Byott and L. Vendramin)}, 
Journal of Combinatorial Algebra Volume 2, Issue 1, (2018) 47-86.

\bibitem{LAA}  
A. Smoktunowicz and A. Smoktunowicz, {\em Set-theoretic solutions of the Yang–Baxter equation and new classes of R-matrices}, 
Linear Algebra and its Applications, Volume 546, 1 June 2018, pages 86--114.

\bibitem{sysak}  
Y.P. Sysak, {\em The adjoint group of radical rings and related questions}, in Ischia Group Theory 2010,
(proceedings of the conference: Ischia, Naples, Italy, 14-17 April 2010), pp. 344–365, World Scientifc, Singapore (2011).

\bibitem{TakSat}
D. Takahashi and J. Satsuma, {\em A soliton cellular automaton}, J. Phys. Soc. Japan 59 (1990) 3514.

\bibitem{Leandro} L. Vendramin, {\em Extensions of set-theoretic solutions of the Yang--Baxter equation and a 
conjecture of Gateva-Ivanova}. Journal of Pure and Applied Algebra, Volume 220, Issue 5 (2016) 1681--2076.

\bibitem{Veselov} 
A.P. Veselov, {\em Yang-Baxter maps and integrable dynamics}, Phys. Lett.  A314 (2003) 214.

\bibitem{Yang}
C.N. Yang, {\em Some exact results for the many-body problem in one dimension with repulsive delta-function interaction}, Rev. Lett. 19 (1967) 1312.
\end{thebibliography}
\end{document}